\documentclass[aps,prx,twocolumn,groupedaddress,nofootinbib]{revtex4-2}

\usepackage{graphicx}
\usepackage{dcolumn}
\usepackage{tikz-cd}
\usepackage{hyperref}
\usepackage{url}
\usepackage{bm}

\usepackage{amsthm,amssymb,amsmath}

\newcommand{\Rnum}[1]{\uppercase\expandafter{\romannumeral #1\relax}}
\usepackage{dsfont}
\usepackage{xcolor}
\usepackage{bbold}
\usepackage{stmaryrd}
\usepackage{subcaption}
\usepackage{physics}
\usepackage{graphicx}

\newcommand{\Acurvy}{\mathcal{A}}
\newcommand{\Acurvybar}{\overline{\mathcal{A}}}

\newcommand{\Bcurvybar}{\overline{\mathcal{B}}}
\newcommand{\Hil}{\mathcal{H}}

\newcommand{\Bcurvy}{\mathcal{B}}
\newcommand{\vecmb}{\vec{m}_b}
\newcommand{\vecmbp}{\vec{m}'_b}
\newcommand{\Ccurvy}{\mathcal{C}}

\newcommand{\Gcurvy}{\mathcal{G}}

\newcommand{\Scurvy}{\mathcal{S}}

\newcommand{\oalg}[1]{\mathfrak{o}_{#1}}

\newtheorem{lemma}{Lemma}
\newtheorem{theorem}{Theorem}
\newtheorem{definition}{Definition}
\newtheorem{proposition}{Proposition}
\newtheorem{corollary}{Corollary}

\begin{document}
	
	\preprint{FILL OUT}
	
	\title{A Bravyi-K\"onig theorem for Floquet codes generated by locally conjugate instantaneous stabiliser groups}
	
	\author{Jelena Mackeprang}
	\email{jelena.mackeprang@gmail.com}
	\author{Jonas Helsen}
	\affiliation{%
		QuSoft \& CWI\\
		  Science Park 123, 1098XG Amsterdam, The Netherlands
	}%
	
	\date{\today}


\begin{abstract}
The Bravyi-K\"onig (BK) theorem is an important no-go theorem for the dynamics of topological stabiliser quantum error correcting codes. It states that any logical operation on a $D$-dimensional topological stabiliser code that can be implemented by a short-depth circuit acts on the codespace as an element of the $D$-th level of the Clifford hierarchy. In recent years, a new type of quantum error correcting codes based on Pauli stabilisers, dubbed Floquet codes, has been introduced. In Floquet codes, syndrome measurements are arranged such that they dynamically generate a codespace at each time step. Here, we show that the BK theorem holds for a definition of Floquet codes based on locally conjugate stabiliser groups. Moreover, we introduce and define a class of generalised unitaries in Floquet codes that need not preserve the codespace at each time step, but that combined with the measurements constitute a valid logical operation. We derive a canonical form of these generalised unitaries and show that the BK theorem holds for them too.
\end{abstract}

\maketitle
\section{Introduction}\label{sec:intro}

Quantum error correction (QEC) and fault tolerant quantum computing are crucial ingredients for universal quantum computation due to the inherent fragility of quantum hardware systems. A rich and well-studied framework is the Pauli stabiliser formalism, where the codespace is chosen to be the joint eigenspace of an abelian subgroup of the Pauli group - the stabiliser group or sometimes simply called the stabiliser - and errors change the sign of one or multiple elements of this group.
 A major bottleneck in designing good (stabiliser) quantum error correcting codes are high-weight measurements, i.e. ones that require addressing multiple qubits at once using possibly entangling operations. This problem motivates the development of codes allowing for higher-weight measurements to be decomposed into smaller ones, such as subsystem codes, where a subset of the logical qubits are turned into gauge qubits. \\
Dynamical codes (or Floquet codes), first introduced in~\cite{hastingsDynamicallyGeneratedLogical2021,haahBoundariesHoneycombCode2022}, pose another way to split up high-weight measurements. Here, one transitions from one stabiliser group to the next by measuring Pauli operators that anti-commute with a subset of stabilisers in the current one. Because the measured Pauli operators each anti-commute with at least one element of the stabiliser, no information is gained about the actual logical state. Logical information is preserved. The stabiliser groups between which one transitions in a Floquet code are dubbed \emph{instantaneous stabiliser groups}, or short ISGs. One can view two consecutive ISGs and the measurements as a subsystem code (as mentioned in ~\cite{fuErrorCorrectionDynamical2025,townsend-teagueFloquetifyingColourCode2023a,guFaulttolerantQuantumArchitectures2025}), but as soon as a third ISG is present, the measurements in a dynamical code do not necessarily define a standard subsystem code measurement sequence. That is because there generally is no consistent logical subspace as for subsystem codes. In periodic dynamical codes, one periodically returns (up to stabiliser signs) to the initial codespace, but the measurement sequences generally introduce a non-trivial logical operation. 
 Other dynamical codes have been introduced in~\cite{davydovaFloquetCodesParent2023,townsend-teagueFloquetifyingColourCode2023a,ellisonFloquetCodesTwist2023,kesselringAnyonCondensationColor2024,davydovaQuantumComputationDynamic2024,haahBoundariesHoneycombCode2022,fuErrorCorrectionDynamical2025}. 
 
 As far as we know, there is no universally agreed-upon definition of a dynamical code. In~\cite{kesselringAnyonCondensationColor2024} and~\cite{davydovaQuantumComputationDynamic2024} the authors define topological dynamical codes as sequences of anyon condensations with respect to a parent anyon theory. Another approach to systematically construct dynamical codes is via the ZX calculus (see~\cite{bombinUnifyingFlavorsFault2024a,townsend-teagueFloquetifyingColourCode2023a,rodatzFloquetifyingStabiliserCodes2024}). A very general approach to dynamical codes is discussed in~\cite{fuErrorCorrectionDynamical2025}, where the authors use a definition of Floquet codes consisting solely of the a priori prescription of the sequential measurements and derive an algorithm to compute different notions of distance. \\ 



In this work, we analyse the dynamical properties of Floquet codes. Namely, we want to know which logical operations can be implemented in a fault-tolerant manner. To do so, we first have to agree on a definition for Floquet codes. We refrain from the very broad notion of dynamical codes as an arbitrary sequence of measurements, as these fit within the broader framework of spacetime codes/spacetime error correction~\cite{delfosseSpacetimeCodesClifford2023}. Instead, we base our definition on the concept of (locally) conjugate stabiliser group, as introduced in~\cite{aasenMeasurementQuantumCellular2023a}. They are the building block of most previously introduced Floquet codes, such as in~\cite{hastingsDynamicallyGeneratedLogical2021,davydovaFloquetCodesParent2023,davydovaQuantumComputationDynamic2024,townsend-teagueFloquetifyingColourCode2023a,ellisonFloquetCodesTwist2023}. To be more precise, we view dynamical codes as a finite sequence of locally conjugate stabiliser groups $\Acurvy_1 \rightarrow \Acurvy_2 \rightarrow \ldots \rightarrow \Acurvy_{\tau}$, where one transitions from $\Acurvy_{t}$ to $\Acurvy_{t+1}$ by measuring specific generators of the respective subsequent stabiliser group, i.e. of $\Acurvy_{t+1}$. Most \emph{topological} dynamical codes introduced so far, such as the ones in~\cite{hastingsDynamicallyGeneratedLogical2021,davydovaFloquetCodesParent2023,davydovaQuantumComputationDynamic2024,kesselringAnyonCondensationColor2024,townsend-teagueFloquetifyingColourCode2023a}, are covered by this definition. 
 However, our work should extend to these more general definitions of dynamical codes in a straightforward manner, as we will discuss in section~\ref{sec:concl} 
 

Having narrowed down a precise definition, we ask whether an important no-go theorem about the dynamics of topological stabiliser codes (TSCs), the Bravyi-K\"onig theorem~\cite{bravyiClassificationTopologicallyProtected2013}, holds for topological Floquet codes. It states that unitary operators on standard $D$-dimensional TSCs that can be implemented by a short-depth, short-range circuit must be in the $D$-th level of the Clifford hierarchy. Here, we consider a sequence of ISGs $\Acurvy_1 \rightarrow \Acurvy_2 \rightarrow \Acurvy_3 \rightarrow \ldots \Acurvy_{\tau}$ each defining a code space at the time steps $t = 1,2,3,4, \ldots$ and the constant-depth, constant-range circuits $U_1$, $U_2$, $U_3$, $\ldots$, $U_{\tau}$ applied at the respective time steps. We ask whether the combined logical operation induced by these unitaries and the measurements is also limited to the Clifford hierarchy.
%

We consider two different types of circuits. First, we only allow unitary operations that preserve the codespace at each time step. We explain why the validity of the BK theorem quickly follows from the definition of Floquet codes based on locally conjugate stabiliser groups.

 Secondly, we move on to a more interesting case --which forms the main result of this paper-- based on the observation that Floquet codes admit a more general type of logical unitary operators, which do not necessarily preserve the codespace at each time step, yet still do  preserve error detectability and logical information. This is possible because there are Pauli operators that do not preserve the codespace, but that nevertheless do not constitute an error because they are immediately absorbed by the subsequent measurement. Within the language of anyon condensation these correspond to anyons that are subsequently condensed. Loosely speaking, we define the following~\emph{generalised logical unitaries in dynamical codes}:
 
 \begin{definition}[Generalised logical unitary in dynamical code (informal)]\label{def:genuninformal}
 	Let $\Acurvy \rightarrow \Bcurvy$ be two subsequent ISGs in a dynamical code. At time step $t$, the state is a joint eigenstate of $\Acurvy$ and at time step $t+1$ it is a joint eigenstate of $\Bcurvy$. A generalised unitary $U$ at time $t$ is a unitary that does not need to preserve the joint eigenspace of $\Acurvy$, but that fulfils the following conditions:
 	
 	\begin{itemize}
 		\item \textbf{Error detectability and self-correction}: After the application of $U$ all previously detectable errors must remain detectable. All errors that were self-correctable must remain self-correctable.
 		\item \textbf{Logical preservation and logical equivalence}: After the application of $U$ and the subsequent measurement logical information must be preserved and the effective logical action must remain the same for all measurement results. 

 	\end{itemize}
 \end{definition}
 
 The formal definition is stated in section~\ref{sec:beyondcodepres}.
  We derive a canonical form (which we believe might be of independent interest) of these generalised logical unitaries in section~\ref{sec:canonicalform} and show theorem~\ref{thm:BKbeyondinformal}:
  
  \begin{theorem}\label{thm:BKbeyondinformal}[Bravyi-K\"onig theorem for generalised logical unitaries (informal)]
  	Let $\Acurvy_0 \rightarrow \Acurvy_1 \rightarrow \ldots \rightarrow \Acurvy_{{\tau}_1}$ be a dynamical code, where each $\Acurvy_t$ is a $D$-dimensional topological stabiliser code. At each time step one applies a constant-depth, constant-range generalised logical unitary as defined in definition~\ref{def:genuninformal}. If $\tau = \mathcal{O}(1)$, then the effective logical operation of all unitaries and measurements amounts to an element in the $D$-th level of the Clifford hierarchy.
  \end{theorem}  
  
  Theorem~\ref{thm:BKbeyondinformal} is stated formally in section~\ref{sec:BKbeyondcodepres}.
 
The paper is structured as follows. In section~\ref{sec:presinfo} we introduce conjugate stabiliser groups, formerly defined in~\cite{aasenMeasurementQuantumCellular2023a}, and elaborate on how they enable the preservation of logical information despite non-trivial projective measurements. We explain why the projection induced by the measurement can be replaced by a unitary transition operator and state one if its representatives. To our knowledge,~\ref{sec:presinfo} constitutes a formal discussion of why logical information can be preserved in dynamical codes that has remained somewhat implicit in previous work. Thus this may be of independent interest for comprehension purposes. In section~\ref{sec:localitytransHH} we introduce the notion of geometric locality in Floquet codes, based on the \emph{locally} conjugate stabiliser groups introduced in~\cite{aasenMeasurementQuantumCellular2023a}. This geometric locality ensures bounded growth of errors within the Floquet sequence and also leads to a unitary transition operator that is of finite depth. The latter is the reason why the BK theorem immediately holds for the previously discussed problem set-up for code-preserving constant-depth unitaries at every time step, which we elaborate on in section~\ref{sec:BKforCodePreserving}. \\
In section~\ref{sec:beyondcodepres} we derive the two general conditions for unitary operations at a specific time step and derive a canonical representative for any unitary that fulfils this condition. Finally, in section~\ref{sec:BKbeyondcodepres} we prove the BK theorem for these generalised unitaries. In the end, in section~\ref{sec:concl} we give a brief discussion and outlook.

 \section{Preservation of information}\label{sec:presinfo}
 
 Here we explain the underlying working principle of many previously introduced stabiliser Floquet codes: Reversible pairs of stabiliser groups. These are pairs of stabiliser groups fulfilling certain anti-commutation relations. A Floquet code is then defined as a sequence of consecutive reversible pairs, where one transitions from one to the next via projective measurements. \newline Reversible pairs of stabiliser groups were formally introduced in in~\cite{aasenMeasurementQuantumCellular2023a}.  
 We define them and discuss how they allow for information preservation in section~\ref{sec:revpairs} and then introduce a unitary operator correctly implementing the dynamics induced by the projective measurements in section~\ref{sec:logicaleffectHH}. Additionally, we discuss how logical operators transform due to the aforementioned projective measurements.

  \subsection{Reversible pairs of stabiliser groups}\label{sec:revpairs}
  We state the definition of a reversible pair of stabiliser groups (taken from~\cite{aasenMeasurementQuantumCellular2023a}) in definition~\ref{def:reversiblepairHH}. For two stabiliser groups $\Acurvy$ and $\Bcurvy$, we use the notation $\Scurvy:=\Acurvy \cap \Bcurvy$ for their intersection (as a set), and $\Acurvybar:=\Acurvy \setminus \Scurvy$ and $\Bcurvybar:=\Bcurvy \setminus \Bcurvybar$. We write $|\Gcurvy|$ for the rank of a stabiliser group $\Gcurvy$.
  \begin{definition}[Reversible pair/Conjugate stabiliser groups]\label{def:reversiblepairHH}
  	Two stabiliser groups $\Acurvy$ and $\Bcurvy$ are called a reversible pair or conjugate if there exist bases (conjugate bases) $\{a_j\} \in \Acurvy \setminus \Scurvy$ and $\{b_j \in \Bcurvy \setminus \Scurvy\}$, such that $a_i$ commutes with all elements of $\{b_j\}$ except for exactly one $b_i\in\{b_j\}$ with which it anti-commutes and vice versa. 
  \end{definition}
  We write $\Acurvy \leftrightarrow \Bcurvy$ for a reversible pair. We specify the signs of the $k$ generators of a stabiliser group $\Gcurvy$ with the help of a vector $\vec{m}_{\Gcurvy} \in \{0,1\}^k$, where the $i$-th generator $g_i$ takes on the eigenvalue $(-1)^{m_{\Gcurvy}^{(i)}}$, with respect to some enumeration of the generators. We denote the corresponding Hilbert space by $\Hil_{\Gcurvy}(\vec{m}_{\Gcurvy})$. We will use the following notation for the projection operator onto the joint eigenspace of a stabiliser group $\Gcurvy$, belonging to the eigenvalues given by $\vec{m}_{\Gcurvy} \in \{0,1\}^k$:
 \begin{equation}\label{eq:PiG}
 	\Pi_{\Gcurvy}(\vec{m}_{\Gcurvy})=	\prod_{j=0}^{k-1} \frac{\mathbb{1}+(-1)^{m^{(j)}_{\Gcurvy}} g_j}{2}.
 \end{equation}
A Floquet code essentially consists of a (possibly periodic) sequence of reversible pairs, i.e.:
\begin{equation}
	\Acurvy_1 \rightarrow \Acurvy_2 \rightarrow \Acurvy_3 \rightarrow \ldots \Acurvy_{\tau},
\end{equation}
where $\Acurvy_{t}\leftrightarrow\Acurvy_{t+1}$. (The sequence may be infinite, as for example is the case for some codes introduced in~\cite{davydovaFloquetCodesParent2023}, but we will assume $\tau$ is finite.) To transition from one stabiliser group to the next, one projects onto the eigenspace of $\Acurvy_{t+1} \setminus (\Acurvy_{t} \cap \Acurvy_{t+1})$. In the following, we will elaborate on why this works and why we purposefully leave out the specification of measurement results. \\

 Take a reversible pair $\Acurvy \leftrightarrow \Bcurvy$ and assume the system is in the joint eigenspace of $\Acurvy$. We now measure/project onto the conjugate basis $\{b_j\}$. We can do so one by one or simultaneously as the measurements commute. We start with some basis element $b$. The update rules for Pauli stabiliser groups demand that all stabilisers in $\Acurvy$ that commute with $b$ remain, while the ones that do not commute with $b$ are removed and replaced by $b$~\cite{hastingsDynamicallyGeneratedLogical2021}. As $\Acurvy$ and $\Bcurvy$ are conjugate, this means that only the one conjugate element $a$ is removed and $b$ added. After repeating this for all operators in $\{b_j\}$, the qubits are in the joint eigenspace of $\Scurvy \cup \Bcurvybar = \Bcurvy$, where $\Bcurvybar := \Bcurvy \setminus \Scurvy$. It is determined by $\vec{m}_{\Bcurvy}$. One thus moves through the stabiliser group sequence up to signs by measuring the respective conjugate bases. 
 
 Note that the projection operator onto the $(-1)^{m^{(i)}_{\Bcurvybar}}$ eigenspace of $\Bcurvybar$, i.e. the one given by~Eqn.~\eqref{eq:PiG} is the same independent of the choice of conjugate basis for $\Bcurvybar$. Thus, the following discussion applies to all conjugate basis choices of $\Acurvybar$ and $\Bcurvybar$ and analogously to all conjugate basis choices for all reversible pairs. We will therefore  always choose the most straight-forward one.

\subsection{Preservation of information}

We will now elaborate on how information is preserved despite the seemingly destructive projective measurements in the Floquet sequence. This is implicitly discussed in the majority of previous work on Floquet codes, e.g.~\cite{hastingsDynamicallyGeneratedLogical2021,davydovaFloquetCodesParent2023,davydovaQuantumComputationDynamic2024}, but this explicit discussion may still be useful for readers unaccustomed to the topic. \\
 
Take any reversible pair $\Acurvy \leftrightarrow \Bcurvy$.   Due to the anti-commutation relations between $\Acurvybar$ and $\Bcurvybar$, the following proposition holds: 

\begin{proposition}\label{prop:pbpapbidHH}
	Let $\Acurvy$ and $\Bcurvy$ be two conjugate stabiliser groups. Let $|\Acurvybar|=|\Bcurvybar|=n_m$. Then, for any $\vec{m}_{\Acurvybar} \in \mathbb{F}_2^{n_m}$ and any $\vec{m}_{\Bcurvybar} \in \mathbb{F}_2^{n_m}$:
	
	\begin{align}\label{eq:pbpapbidHH}
		\Pi_{\Bcurvybar}(\vec{m}_{\Bcurvybar}) \Pi_{\overline{\Acurvy}}(\vec{m}_{\overline{\Acurvy}}) \Pi_{\Bcurvybar}(\vec{m}_{\Bcurvybar}) = \frac{1}{2^{n_m}} \Pi_{\overline{\Bcurvy}}(\vec{m}_{\overline{\Bcurvy}}),
	\end{align}
	and:
	\begin{align}\label{eq:papbpaidHH}
		\Pi_{\Acurvybar}(\vec{m}_{\Acurvybar}) \Pi_{\Bcurvybar}(\vec{m}_{\Bcurvybar}) \Pi_{\Acurvybar}(\vec{m}_{\Acurvybar}) = \frac{1}{2^{n_m}} \Pi_{\Acurvybar}(\vec{m}_{\Acurvybar}).
	\end{align}
	
\end{proposition}
We prove proposition~\ref{prop:pbpapbidHH} in section~\ref{sec:remainingproofs}. From this it straightforwardly follows that all measurement results $\vec{m}_{\overline{\Bcurvy}} \in \mathbb{F}_2^{n_m}$ are equally as likely:

\begin{corollary}\label{cor:probsameHH}
	Let $\Acurvy \leftrightarrow \Bcurvy$ be a reversible pair, as defined in definition~\ref{def:reversiblepairHH}. Let $n_m:=|\Acurvybar|=|\Bcurvybar|$. Then, for any $\vec{m}_{\Acurvy} \in \mathbb{F}_2^{n_m}$ and any $\vec{m}_{\Bcurvybar} \in \mathbb{F}_2^{n_m}$:
	\begin{equation}\label{eq:HHprobsame}
		\bra{\Psi} \Pi_{\Bcurvybar}(\vec{m}_{\Bcurvybar}) \ket{\Psi}=1/2^{n_m},
	\end{equation}
for all $\ket{\Psi} \in \Hil_{\Acurvy}(\vec{m}_{\Acurvy})$.
 \end{corollary} 
 Corollary~\ref{cor:probsameHH} is the reason for why logical information is preserved. To be more precise, we show lemma~\ref{lem:orthpres} below, in which we already leave out the specification of $\vec{m}_{\Acurvybar}$ and $\vec{m}_{\Bcurvybar}$ from our notation.
 \begin{lemma}\label{lem:orthpres}
 	For any $\ket{\Psi},\ket{\Phi} \in \Hil_{\Acurvy}$, the projection $\Pi_{\Bcurvybar}$ preserves the scalar product up to the normalisation factor $\frac{1}{2^{n_m}}$:
 \begin{equation}\label{eq:scalarpres}
 	\bra{\Phi} \ket{\Psi} = 2^{n_m}\bra{\Phi}\Pi_{\Bcurvybar} \Pi_{\Bcurvybar} \ket{\Psi}
 \end{equation}
 and can thus (up to the normalisation factor) be replaced by a unitary $K_{\Acurvy,\Bcurvy}: \Hil_{\Acurvy} \rightarrow \Hil_{\Bcurvy}$, i.e.:
 \begin{equation}
 	2^{\frac{n_m}{2}}\Pi_{\Bcurvy}\Pi_{\Acurvy}=K_{\Acurvy,\Bcurvy}\Pi_{\Acurvy}.\label{eq:Kabfirst}
 \end{equation}
 \end{lemma}
 \begin{proof}
 Eqn.\eqref{eq:scalarpres} follows from proposition~\ref{prop:pbpapbidHH}. A linear map from a finite Hilbert space to another finite Hilbert space of the same dimension that preserves the scalar product is a unitary. \footnote{Technically, the left and right hand side of Eqn.~\eqref{eq:Kabfirst} are defined on different domains and co-domains. There, strictly speaking, is a difference between an operator with domain $\Hil_{\Acurvy}$ and co-domain $\Hil_{\Bcurvy}$ and one that acts on the entire Hilbert space $\Hil^{\otimes n}$, but maps the linear subspace $\Hil_{\Acurvy}$ to the linear subspace $\Hil_{\Bcurvy}$, while mapping the orthogonal complement $\Hil_{\Acurvy}^{\perp}$ to zero. However, we will treat these two cases as interchangeable in the following, as it will be clear what we mean.}
 \end{proof}
 This implies that, due to the anti-commutation relations, the measurement does not induce an irreversible action. It is precisely because we do not gain any actual information (as the measurement results are completely random) that the projection is proportional to a unitary. One can view these transitions as a specific form of code switching. Within the framework of stabiliser tableaus~\cite{aaronsonImprovedSimulationStabilizer2004}, the measurements amount to measuring elements of the destabiliser. We would like to point out that transitioning between stabiliser groups in this way has previously been discussed in~\cite{colladayRewiringStabilizerCodes2018} without calling it ``dynamical codes" or ``Floquet codes". \\
 
 As all measurement results $\vec{m}_{\Bcurvybar}$ are equally as likely and independent from $\vec{m}_{\Acurvybar}$, we will leave out the specification of either from now on, assume both are fixed and often not speak of ``measurements" but ``projections". All our results will hold for all possible measurement results.

Another way to see that logical information is preserved is by looking at the logical Pauli operators of each individual code in the Floquet sequence. Recall that the logical Pauli operators of a stabiliser code are determined by the non-trivial elements of its normaliser $\mathcal{N}(\Acurvy)$ in the $n$-qubit Pauli group $\mathcal{P}(n)$ which is defined as:
\begin{equation}\label{eq:normaliser}
	\mathcal{N}(\mathcal{A})=\{P \in \mathcal{P}(n): Ps =sP \hspace{1em} \forall\, s \in \Acurvy\},
\end{equation}
i.e. the Pauli operators that commute element-wise with $\Acurvy$. For any $[n,k]$-stabiliser code the quotient group $\mathcal{N}(\Acurvy) / \Acurvy$ is isomorphic as an algebra to $\mathcal{P}(k)$~\cite{gottesmanIntroductionQuantumError2009}. An element of $\mathcal{N}(\Acurvy)$ may commute with all elements of $\Bcurvybar$. It would thus be unaffected by the measurements. It may, however, also anti-commute with some operators $b \in \Bcurvybar$. Thus, one might think we obtain some information about its value through the measurement. However, that is not the case. To explain why, we restate proposition 2.1(d) from~\cite{aasenMeasurementQuantumCellular2023a}:

\begin{proposition}[Shared logical operators]\label{prop:LsharedHH}
	Let $\Acurvy \leftrightarrow \Bcurvy$ be a reversible pair. Then, there exists a group of Pauli operators $L$, such that $L\Acurvy$ is the set of all logical Pauli operators for $\Bcurvy$ and $L\Bcurvy$ is the set of all logical Pauli operators for $\Bcurvy$.
\end{proposition}
Proposition~\ref{prop:LsharedHH} tells us that for any $Q \in \mathcal{N}(\Acurvy)$, there exists an equivalent $Q'=sQ$, where $s$ is some element of $\Acurvy$, such that $Q' \in \mathcal{N}(\Bcurvy)$.
For any $\ket{\Psi} \in \Hil_{\Acurvy}$ the expectation value of $Q$ will be the same as the one of $Q'=sQ$, because $s$ acts trivially on $\ket{\Psi}$. The expectation value of $Q'$ (or equivalently of $Q$) after the measurement is $2^{n_m}	\bra{\Psi} \Pi_{\Bcurvybar}Q'\Pi_{\Bcurvybar}\ket{\Psi}$. (We already normalised the resulting state after the measurement.)
However, due to proposition~\ref{prop:pbpapbidHH} and the fact that $Q'$ commutes with all of $\Bcurvy$ (and thus also $\Pi_{\Bcurvybar}$) and $\Acurvy$ we can write:
\begin{align}
	2^{n_m}	\bra{\Psi} \Pi_{\Bcurvybar}Q'\Pi_{\Bcurvybar}\ket{\Psi}&=	2^{n_m}	\bra{\Psi} \Pi_{\Bcurvybar} Q'\ket{\Psi}\\
	&=2^{n_m}	\bra{\Psi} \Pi_{\Acurvybar} \Pi_{\Bcurvybar}  Q'\Pi_{\Acurvybar} \ket{\Psi} \label{eq:expectHH2} \\
		&= 2^{n_m}	\bra{\Psi} \Pi_{\Acurvybar} \Pi_{\Bcurvybar}    \Pi_{\Acurvybar}Q'\ket{\Psi} \label{eq:expectHH3} \\
	&= 	\bra{\Psi} \Pi_{\Acurvybar}  Q'\ket{\Psi} \\
 	&=\bra{\Psi}   Q'\ket{\Psi}\\
	&=\bra{\Psi}   Q\ket{\Psi},
\end{align}
where we inserted the projection operator $\Pi_{\Acurvybar}$ in Eqn.~\eqref{eq:expectHH2} on both sides and used proposition~\ref{prop:pbpapbidHH} in Eqn.~\eqref{eq:expectHH3}.
In summary, for any logical Pauli operator $Q$ in $\mathcal{N}(\Acurvy)$, there exists a representative $Q'$ that is also an element of $\mathcal{N}(\Bcurvy)$ and whose expectation value remains the same after the projection onto the eigenspace of $\Bcurvybar$. Thus, no information is gained about any of the logical Pauli operators.

  \subsection{Logical effect of the Floquet projections} \label{sec:logicaleffectHH}
  
In this section we formalise the unitary evolution induced by the projection operator and discuss how it transforms the logical operators of the codespace. We formally define the unitary operator that, up to the normalisation factor, acts equivalently on a $\ket{\Psi} \in \Hil_{\Acurvy}$ as the projection operator $\Pi_{\Bcurvybar}$, where $\Acurvy \leftrightarrow \Bcurvy$. We call this operator the ``Floquet transition operator" or just the ``transition operator". We then discuss how this operator acts on the logical operators of $\Hil_{\Acurvy}$.

 \subsubsection{Floquet transition operator}  \label{sec:floquettransitionop}
 
 As shown in lemma~\ref{lem:orthpres}, the projection $\Pi_{\Bcurvybar}$ can be replaced by the Floquet transition operator:
 
 \begin{definition}[The transitions operator $K_{\Acurvy,\Bcurvy}$]\label{def:KdefHH}
 	Let $\Acurvy \leftrightarrow \Bcurvy$ be a reversible pair and let the system be in the state $\ket{\Psi} \in \Hil_{\Acurvy}$. Now project via $\Pi_{\Bcurvybar}$. The unitary operator $K_{\Acurvy,\Bcurvy}$ describes the effect of $\Pi_{\Bcurvybar}$:  
 	\begin{equation}\label{eq:HHK}
 		K_{\Acurvy,\Bcurvy}: \Hil_{\Acurvy} \rightarrow \Hil_{\Bcurvy}: \ket{\Psi} \mapsto 2^{n_m/2} \Pi_{\Bcurvybar} \ket{\Psi}.
 	\end{equation}
 \end{definition}
The transition operator generally has a non-trivial effect on the codespace that depends on the encoding. For a periodic Floquet codes, e.g. $\Acurvy \rightarrow \Bcurvy \rightarrow \Ccurvy \rightarrow \Acurvy$, one can obtain the collective effect of the transition operators by following the evolution of fixed representatives of logical operators in $\Acurvy$. In the case of Floquet codes based on ayon condensation, the logical effect is given by the automorphism of the anyon theory that the sequence of measurements (condensations) induces~\cite{kesselringAnyonCondensationColor2024,davydovaQuantumComputationDynamic2024}.
 There is an easy expression for the inverse of $K_{\Acurvy,\Bcurvy}$, which we prove in appendix (section~\ref{sec:remainingproofs}).
 \begin{lemma} \label{lem:KunitaryHH}
 	$K_{\Bcurvy,\Acurvy}: \Hil_{\Bcurvy} \rightarrow \Hil_{\Acurvy}$ is the inverse of $K_{\Acurvy,\Bcurvy}$, i.e. $K_{\Bcurvy,\Acurvy}K_{\Acurvy,\Bcurvy}= \mathbb{1}_{\Acurvy}$, where $\mathbb{1}_{\Acurvy}$ denotes the identity on $\Hil_{\Acurvy}$.
 \end{lemma}
 As discussed in~\cite{aasenMeasurementQuantumCellular2023a} and~\cite{davydovaFloquetCodesParent2023}, we can find a unitary $V_{\Acurvy,\Bcurvy}: \Hil^{\otimes n} \rightarrow \Hil^{\otimes n}$ acting on the \emph{entire} Hilbert space that implements this transformation.
 \begin{proposition}\label{prop:VabimplementsKab}
 	Let $\Acurvy \leftrightarrow \Bcurvy$ be a reversible pair, $\Scurvy= \Acurvy \cap \Bcurvy$ and $\{a_i\}$ and $\{b_j\}$ be the conjugate basis elements, where with respect to some arbitrary enumeration $a_i$ anti-commutes with $b_i$ and commutes with all other $b_j$. The unitary operator $V_{\Acurvy,\Bcurvy}: \Hil^{\otimes n} \rightarrow \Hil^{\otimes n}$ given by:
 	\begin{equation}
 		V_{\Acurvy,\Bcurvy} = \prod_{i} \left(\frac{a_i+b_i}{\sqrt{2}}\right)
 	\end{equation}
 	implements $K_{\Acurvy,\Bcurvy}$ defined in definition~\ref{def:KdefHH}, i.e.:
 	\begin{equation}\label{eq:VabimplementsK}
 			V_{\Acurvy,\Bcurvy} \Pi_{\Acurvy}= K_{\Acurvy,\Bcurvy} \Pi_{\Acurvy}
 	\end{equation}
 	and 
 	\begin{equation}\label{eq:VabimplementsKdag}
 	\Pi_{\Acurvy}	V^{\dag}_{\Acurvy,\Bcurvy} = \Pi_{\Acurvy}K^{\dag}_{\Acurvy,\Bcurvy}. 
 	\end{equation}
 	Thus, $V_{\Acurvy,\Bcurvy}$ maps $\Hil_{\Acurvy}$ to $\Hil_{\Bcurvy}$ and also maps the orthogonal complement $\Hil_{\Acurvy}^{\perp}$ to $\Hil_{\Bcurvy}^{\perp}$.
 \end{proposition}
 
 \begin{proof}
 	Inserting $V_{\Acurvy,\Bcurvy}$ into Eqn.~\eqref{eq:VabimplementsK} results in:
 	\begin{align}
 		V_{\Acurvy,\Bcurvy} \Pi_{\Acurvy}&=	V_{\Acurvy,\Bcurvy} \Pi_{\Scurvy}\Pi_{\Acurvybar}\\
&=\prod_{i} \left(\frac{a_i+b_i}{\sqrt{2}}\right) \prod_{j}\left(\frac{\mathbb{1}+a_j}{2}\right)\Pi_{\Scurvy}\\
&=\prod_{i} \left(\frac{\mathbb{1}+b_i}{\sqrt{2}}\right)\left(\prod_{j}\frac{\mathbb{1}+a_j}{2}\right)\Pi_{\Scurvy}\label{eq:Vababsorb}\\
&= 2^{\frac{n_m}{2}}\Pi_{\Bcurvybar}\Pi_{\Acurvy}\label{eq:VabKdef}\\
&=K_{\Acurvy,\Bcurvy} \Pi_{\Acurvy},
 	\end{align}
 	where we absorbed the stabilisers in $\Acurvybar$ into the projection operators in Eqn.~\eqref{eq:Vababsorb}. This is possible as each $b_i$ only anti-commutes with exactly one $a_i$ and thus the operator containing $b_i$ can be pulled through all projection operators containing $a_j \neq a_i$. We inserted the definition of $K_{\Acurvy,\Bcurvy}$ in Eqn.~\eqref{eq:VabKdef}. Eqn.~\eqref{eq:VabimplementsKdag} follows analogously.
 \end{proof}

\subsubsection{Transformation of logical operations}
The BK theorem is a no-go theorem for logical operators on quantum error correcting codes. In this section we define logical operations on stabiliser codes and discuss the transformation they undergo induced by the projection. 

\begin{definition}[Logical operations on stabiliser codes]
	Let $\Acurvy \subseteq \mathcal{P}_n$ define a stabiliser code on $n$ physical qubits. Then $\oalg{\Acurvy}$ are the unitary operators in $U(\Hil^{\otimes n})$ that preserve the codespace:
	\begin{equation}
		\oalg{\Acurvy} := \{U \in U(\Hil^{\otimes n}): \Pi_{\Acurvy} U \Pi_{\Acurvy}=U\Pi_{\Acurvy}\}.
	\end{equation}.
\end{definition}
The set $\oalg{\Acurvy}$ is the set of logical operations on a QEC that is implicitly or explicitly referenced in most of the work considering the dynamics of QEC, such as~\cite{eastinRestrictionsTransversalEncoded2009,bravyiClassificationTopologicallyProtected2013}. More explicitly, every $U \in \oalg{\Acurvy}$ can be written as the orthogonal decomposition:
\begin{equation} \label{eq:Udecomp}
	U=U_{\Acurvy} \oplus U^{\perp}_{\Acurvy},
\end{equation}
where $U_{\Acurvy}=\Pi_{\Acurvy} U_{\Acurvy}\Pi_{\Acurvy}$ is the part of $U$ acting on the codespace and $U^{\perp}$ is its orthogonal complement. 
For each $U  \in \oalg{\Acurvy}$, the projection $\Pi_{\Bcurvybar}$ induces the transformation of the part $U_{\Acurvy}$ in Eqn.~\eqref{eq:Udecomp}, i.e. the map:
\begin{equation}
 U_{\Acurvy} \mapsto	K_{\Acurvy,\Bcurvy} U_{\Acurvy} K_{\Acurvy,\Bcurvy}^{\dag}. 
\end{equation}
 One can now obtain all representatives of this logical operation by adding some orthogonal complement $U_{\Bcurvy}^{\perp} \in U(\Hil_{\Bcurvy}^{\perp})$. One such representative can be obtained via $V_{\Acurvy,\Bcurvy}$:
 \begin{align}
 V_{\Acurvy,\Bcurvy} U V^{\dag}_{\Acurvy,\Bcurvy} &= V_{\Acurvy,\Bcurvy} (U_{\Acurvy}+U^{\perp}_{\Acurvy}) V^{\dag}_{\Acurvy,\Bcurvy} \\
 &=K_{\Acurvy,\Bcurvy} U_{\Acurvy}K^{\dag}_{\Acurvy,\Bcurvy} +U_{\Bcurvy}^{\perp},
 \end{align}
 where $U^{\perp}_{\Bcurvy}$ is some operator acting on $\Hil_{\Bcurvy}^{\perp}$ only and we used proposition~\ref{prop:VabimplementsKab}.\\
 
  \paragraph*{A note on error correction:}
 One has to differentiate between the measurements inducing the transition $\Acurvy \rightarrow \Bcurvy$ and the measurements actually collecting syndrome information. The results of the conjugate basis measurements do not reveal any error information. Thus, in previously introduced Floquet codes, they do not only measure the conjugate basis, but additional stabilisers in $\Acurvy \cap \Bcurvy$. This measurement is redundant in the absence of errors, as it is already a shared stabiliser of $\Acurvy$ and $\Bcurvy$. Its measurement is what allows for clever read-out of error syndromes in spacetime, as elaborated in~\cite{hastingsDynamicallyGeneratedLogical2021,davydovaFloquetCodesParent2023}. However, as we here assume error-free dynamics, we will leave out the redundant measurement. We consider only the minimal set of measurements necessary to project onto the eigenspace of $\Bcurvybar$.

 \section{Floquet codes defined by locally conjugate stabiliser groups}\label{sec:localitytransHH}

 In this section we add the notion of locality to conjugate stabiliser groups, following the definitions introduced in~\cite{aasenMeasurementQuantumCellular2023a}. We define Floquet codes based on \emph{locally conjugate stabiliser groups}, briefly comment on how this definition ensures a locally bounded growth of errors and explain why for locally conjugate stabiliser groups $\Acurvy$ and $\Bcurvy$ the operator $V_{\Acurvy,\Bcurvy}$ defined in proposition~\ref{prop:VabimplementsKab} can be implemented by a constant-depth circuit. \\
 
 For simplicity's sake we will discuss Floquet codes on $D$-dimensional hypercubes. However, for other types of lattices $\Lambda'$, as long as there exists some transformation from $\Lambda'$ to the hypercube that does not let distances between points grow at more than constant rate, all arguments, analogously to~\cite{bravyiClassificationTopologicallyProtected2013}, still hold. We need to specify the notion of a (connected) region $R$ on a lattice $\Lambda$, which will just be a set of (connected) points.
 As a distance measure $d_{\Lambda}$ we utilise the normal Euclidean distance and by the diameter $\mathrm{diam}(R)$ of a region $R$ we mean the maximum distance between any two points in the region. 
 
 We say a region $R$ is $l$-local for an integer $l$ if its diameter is less or equal to $l$. Moreover, we formally introduce the $\rho$-neighbourhood of a region $R$:
 
 \begin{definition}[$\rho$-neighbourhood of a region $R$: $\mathcal{B}_{\rho}(R)$]
 	The $\rho$-neighbourhood $\mathcal{B}_{\rho}(R)$ of a region $R$ on a lattice $\Lambda$ is the union of $R$ and all points $p \in \Lambda$ that fulfil:
 \begin{equation}
 		\min_{q \in R}d_{\Lambda}(p,q) \leq \rho.
 	\end{equation} 
 \end{definition}
 Roughly speaking, it is the collection of points $\rho$ ``away" from $R$. Note that the $\rho$-neighbourhood of $R$ includes $R$. \newline
 Having defined these two notions, we define $l$-locally generated subgroups (taken from~\cite{aasenMeasurementQuantumCellular2023a}) in definition~\ref{def:locallygensubgroups}. Recall that for a reversible pair $\Acurvy \leftrightarrow \Bcurvy$, we use the notation $\Scurvy:=\Acurvy \cap \Bcurvy$ and $\Acurvybar:=\Acurvy \setminus \Scurvy$ and $\Bcurvybar:= \Bcurvy \setminus \Scurvy$. We additionally define some new notation. The \emph{support} of a Pauli operator $P$, denoted by $\mathrm{supp}(P)$, is the set of vertices on the lattice $\Lambda$ on which $P$ acts non-trivially.   
 
 \begin{definition} [$l$-locally generated subgroups]\label{def:locallygensubgroups}
 	An operator $s \in \mathcal{P}$ acting on qubits positioned at the vertices of a lattice $\Lambda$ is $l$-local, for an integer $l >0$, if the diameter of its support fulfils $\mathrm{diam}(\mathrm{supp}(s))\leq l$  A stabiliser group $\Acurvy$ is $l$-locally generated above a subgroup $\Gcurvy \subset \Acurvy$ of $\Acurvy$, if there exists a set $\{a_i \in \Acurvy \setminus \Gcurvy\}$ of $l$-local operators such that $\Acurvy$ is generated by $\{a_i\in \Acurvy \setminus \Gcurvy\}$ and $\Gcurvy$.
 \end{definition}
Note that there is a subtle difference between the size of the support of $s$ and its diameter. We here require geometric locality, i.e. the support of $s$ really must be contained within a connected region of small diameter. We will often say ``$s$ is supported on a region $R$" to specify that its support is contained entirely within $R$. \newline
 We now paraphrase definition \Rnum{2}.3 from~\cite{aasenMeasurementQuantumCellular2023a}.
 \begin{definition} [$l$-local reversibility] \label{def:conjlocal}
 	Let $\Acurvy$ and $\Bcurvy$ be two stabiliser groups acting on qubits positioned at the vertices of a lattice $\Lambda$. Let $\Acurvy \leftrightarrow \Bcurvy$ be a reversible pair according to definition~\ref{def:reversiblepairHH}. $\Acurvy \leftrightarrow \Bcurvy$ is a \textbf{locally reversible pair} if $\Acurvy$ and $\Bcurvy$ are $l$-locally generated above $\Scurvy$ by the $l$-local conjugate bases $\{a_i \in \Acurvybar\}$ and $\{b_i \in \Bcurvy \setminus \Bcurvybar\}$.
 \end{definition}
 Local reversibility leads to preservation of locality for Pauli operators (and thus Pauli errors). We summarise this in a slightly generalised version of corollary \Rnum{2}.5 from~\cite{aasenMeasurementQuantumCellular2023a}:
 \begin{proposition}\label{prop:Psurvives}
 	Let $\Acurvy \subset \mathcal{P}_n$ and $\Bcurvy\subset \mathcal{P}_n$ be $l$-conjugate stabiliser groups defined on a lattice $\Lambda$ with conjugate bases $\{a_i\}$ and $\{b_i\}$. For any $P \in \mathcal{P}_n$ supported only on a a finite region $R$, there exists an $a \in \Acurvybar$ supported only on $\mathcal{B}_{2l}(R)$, such that $aP$ commutes with $\Bcurvybar$ element-wise.
 \end{proposition}
 \begin{proof}
 	Take a conjugate basis for $\Acurvy \leftrightarrow \Bcurvy$. Check which basis elements $b_{j} \in \{b_i\}$ $P$ anti-commutes with. They are supported only on $\mathcal{B}_{l}(R)$, as $\{b_i\}$ is $l$-local. Multiply $P$ with the respective conjugate basis elements $a_{i}$. These are supported on $\mathcal{B}_{2l}(R)$.  The resulting product commutes with $\Bcurvy \setminus \Scurvy$ element-wise. 
 \end{proof}
 Local conjugation ensures that the projective measurement maps local errors at time step $t$ to local errors at time step $t+1$ (as stated in proposition~\ref{prop:Psurvives}). 
 This seems to be the working principle of error correction in most Floquet codes introduced so far, such as~\cite{hastingsDynamicallyGeneratedLogical2021,davydovaFloquetCodesParent2023,davydovaQuantumComputationDynamic2024,townsend-teagueFloquetifyingColourCode2023a,ellisonFloquetCodesTwist2023a}. The code introduced in~\cite{vuillotPlanarFloquetCodes2021} seems to not have $l$-local conjugate bases and contains small local errors that turn into logical operators after the measurement, possibly caused by the lack of $l$-conjugate bases. We conjecture that the conjugation property in itself enables information preservation, while $l$-local conjugation enables delayed error correction. We do not aim here at formally discussing error correction. (See~\cite{kesselringAnyonCondensationColor2024,townsend-teagueFloquetifyingColourCode2023,davydovaQuantumComputationDynamic2024} for a discussion of error correction for Floquet codes within the framework of error detectors, or~\cite{fuErrorCorrectionDynamical2025} for the discussion of distance for dynamical codes generated by arbitrary measurement sequences. A general discussion of fault tolerance in time for Clifford circuits is given in~\cite{delfosseSpacetimeCodesClifford2023}.) The line of thought above justifies defining Floquet codes as a sequence of $l$-local reversible pairs:
 \begin{definition}[Floquet code $\mathcal{F}$]\label{def:floquet}
 	Let $\Acurvy_0$, $\Acurvy_1$, $\ldots$ be a possibly infinite/possibly periodic sequence of $k$-dimensional stabiliser groups $\subseteq \mathcal{P}(n)$, where each consecutive pair is an $l$-locally conjugate reversible pair. Let $\mathcal{H}_t$ be the joint eigenspace of all elements in $\Acurvy_t$. A Floquet code $\mathcal{F}$ is the sequence of Hilbert spaces $\mathcal{H}_t$, where one arrives at the next Hilbert space $\mathcal{H}_{t+1}$ through a projective measurement of the respective conjugate basis elements. We write $\Hil_0 \rightarrow \Hil_1 \rightarrow \ldots$. We call the Floquet code an $[n,k]$-Floquet code.
 \end{definition}
 We note again that this definition does not cover everything called a ``Floquet code" in the literature, but we believe it is sufficiently expansive to be of use. 
 Sometimes we will simply call the sequence $\Acurvy_0$, $\Acurvy_1$, $\ldots$ the Floquet code $\mathcal{F}$. Note that a Floquet code is not simply the sequence of Hilbert space $\Hil_t$, but that we include the \emph{dynamics} by moving through the Hilbert spaces via projective measurements. We now chose all bases to be $l$-conjugate with $l$ being the same for all pairs, but one could have also chosen a different $l_{t,t+1}$ for each pair. However, to upper bound the locality of transformed errors or logical Paulis, one could then simply use the largest $l$. We will thus, for sake of simplicity, work with one $l$, as all of the $l$s will be assumed to be constant with respect to the qubit number $n$ anyway.
 Note that due to the definition of reversible pairs (see definition~\ref{def:reversiblepairHH}), every stabiliser group in the sequence is of the same size. As pointed out in~\cite{davydovaFloquetCodesParent2023,aasenMeasurementQuantumCellular2023a}, the projection operation in Floquet codes defined by locally conjugate stabiliser groups can be replaced by short-depth circuit.
 
 \begin{corollary}\label{cor:Vablocal}
 	Let $\Acurvy \leftrightarrow \Bcurvy$ be a $l$-locally conjugate pair of stabiliser groups and $\{a_i\}$ and $\{b_j\}$ the respective locally conjugate bases. Then, $V_{\Acurvy,\Bcurvy}$ defined in proposition~\ref{prop:VabimplementsKab} implements the transition induced by the projection onto $\Pi_{\Bcurvybar}$ and is a constant-depth, finite range circuit.
 \end{corollary}
 \begin{proof}
	$V_{\Acurvy,\Bcurvy}$ is given by:
	\begin{equation}
	 V_{\Acurvy,\Bcurvy} = \prod_{i} \left(\frac{a_i+b_i}{\sqrt{2}}\right).
	 \end{equation}
	By definition, this is a product of commuting operators. As the support of each $a_i$ and $b_i$ and thus the support of $\frac{a_i+b_i}{\sqrt{2}}$ is locally bounded, $V_{\Acurvy,\Bcurvy}$ is a product of geometrically local commuting operators, which can be implemented by a constant-depth finite-range circuit~\cite{davydovaFloquetCodesParent2023}. 
\end{proof}
 
 \section{The BK theorem for code-preserving unitaries}\label{sec:BKforCodePreserving}
 
 Here, we briefly elaborate on the validity of the BK theorem for code-preserving unitaries. The BK theorem is  a no-got theorem for the dynamics of $D$-dimensional topological stabiliser codes (TSCs). The code-space of a topological stabiliser code is the ground space of a Hamiltonian whose eigenspaces define topological phases. We utilise the same rough definition as in~\cite{bravyiClassificationTopologicallyProtected2013}, namely that a stabiliser code is topological if its stabilisers are defined with respect to qubits placed on a $D$-dimensional lattice of length $\lambda$ and its distance is $\mathcal{O}(\lambda^{\frac{1}{D}})$. \newline
 Let us first restate the original BK theorem (in its general form) from~\cite{bravyiClassificationTopologicallyProtected2013}. Given two stabiliser groups $\Acurvy$ and $\Bcurvy$, the authors of~\cite{bravyiClassificationTopologicallyProtected2013} call a unitary $U \in U(\Hil^{\otimes n})$ a \emph{morphism} between $\Acurvy$ and $\Bcurvy$ if it maps the codespace of $\Acurvy$ to the codespace of $\Bcurvy$:
 
 \begin{equation}
 	U\Pi_{\Acurvy}U^{\dag}=\Pi_{\Bcurvy}.
 \end{equation}
 \begin{theorem}[Original BK theorem]\label{thm:BKog}
 If $U$ is a morphism between $D$-dimensional topological stabiliser codes $\Acurvy$ and $\Bcurvy$ and can be implemented by a constanct-depth, finite-range circuit, then $U$ implements an element of the $D$-th level of the Clifford hierarchy.
 \end{theorem}
 Corollary~\ref{cor:Vablocal} implies that theorem~\ref{thm:BKog} holds for Floquet codes where at each time step one applies a code-preserving constant-depth unitary.
\begin{theorem}[BK theorem for Floquet codes defined by locally conjugate stabiliser codes] \label{thm:BKcodepreserving}
 	Let $\Acurvy_1 \rightarrow \Acurvy_2 \rightarrow \ldots \rightarrow \Acurvy_{\tau}$ be a Floquet code defined in definition~\ref{def:floquet}, where $\tau$ is $\mathcal{O}(1)$ with respect to the number of physical qubits $n$. Let $\{U_t\}$ be an indexed set of size $\tau$ containing constant depth, constant range circuits.
 	Let each individual stabiliser code in the Floquet code be a $D$-dimensional TSC. 
 	Apply $U_t$ at time step $t$ before the projective measurement moving the state to $\mathcal{H}_{t+1}$. The resulting logical operation will be in the $D$-th level of the Clifford hierarchy, as long as $\tau$, $l$ and each $h_t$ and $r_t$ is $\mathcal{O}(1)$.
 \end{theorem}
 \begin{proof}
 	Let $V_{\Acurvy_{t},\Acurvy_{t+1}}$ be the transition unitary defined in proposition~\ref{prop:VabimplementsKab} for two consecutive stabiliser groups $\Acurvy_t$ and $\Acurvy_{t+1}$. We define  	$V_{\Acurvy_{-1},\Acurvy_{0}} := \mathbb{1}$. The resulting operation given by the unitaries and the projections is:
 	\begin{equation}\label{eq:UtVab}
 	U_{\mathrm{total}}	 = \prod_{i=0}^{\tau} U_{\tau-i}V_{\Acurvy_{\tau-i-1},\Acurvy_{\tau-i}}.
 	\end{equation}
  The operator $U_{\mathrm{total}}$ is a morphism between $\Acurvy_1$ and $\Acurvy_{\tau}$. Because of corollary~\ref{cor:Vablocal} it can be implemented by a constant-depth, finite-range circuit. Hence, the original BK theorem applies.
 \end{proof}
 
 In summary, the BK theorem is fulfilled for our definition of Floquet codes, which also encompasses most previously introduced dynamical codes.

 \section{Logical operations beyond code preservation}\label{sec:beyondcodepres}
 
As mentioned in section~\ref{sec:intro}, there is additional freedom when it comes to logical unitaries in dynamical codes. Consider an arbitrary reversible pair $\Acurvy \leftrightarrow \Bcurvy$ and let $\{a_i\}$ and $\{b_j\}$ be the conjugate bases. If $U_{\Acurvy}$ is a code preserving logical unitary with respect to $\Acurvy$ and one applies the operation $bU_{\Acurvy}$, where $b \in \{b_j\}$, right before the projection $\Pi_{\Bcurvybar}$, one has essentially performed $U_{\Acurvy}$. This leads to the question if there are more general unitary operations that are compatible with quantum error correction and the preservation of information, with which one could circumvent the BK theorem for Floquet codes. Certainly, in this case, our earlier argument fails, as one cannot construct the resulting logical operation as in Eqn.~\eqref{eq:UtVab}, i.e. as a product of the constant-depth unitaries $U_t$ and the transition operators, if the unitaries $U_t$ are not code-preserving. That is because then $\Pi_{\Bcurvybar}  U\Pi_{\Acurvy} \neq \Pi_{\Bcurvybar}\Pi_{\Acurvy} U \Pi_{\Acurvy}$ and we cannot immediately replace $\Pi_{\Bcurvybar} \Pi_{\Acurvy}$ by a unitary (up to the normalisation factor). 

 Here we derive what we consider sensible conditions for generalised (i.e. non-codespace preserving) logical unitary operations in dynamical codes. In definition~\ref{def:genunitaries}, we will define a generalised logical unitary for a Floquet transition as a unitary fulfilling specific conditions. We elaborate on these  conditions within sections~\ref{sec:errordetect} and~\ref{sec:logicalpres}. Then, in section~\ref{sec:canonicalform}, we will derive a canonical form for these operations. Using this canonical form we show that the BK theorem holds for generalised logical unitaries in section~\ref{sec:BKbeyondcodepres}.
 
 \begin{definition}\label{def:genunitaries}
 	Let $\Acurvy \leftrightarrow \Bcurvy$ be a reversible pair of stabiliser groups with the conjugate bases $\{a_i\}$ and $\{b_i\}$ spanning $\Acurvybar$ and $\Bcurvybar$ respectively. Let $\Acurvy \rightarrow \Bcurvy$ be their associated Floquet transition. We define a \textbf{generalised logical unitary} between $\Acurvy$ and $\Bcurvy$ as a $U \in U(\Hil^{\otimes n})$ fulfilling the following conditions for all $\vec{m}_b$, $\vec{m}'_b$ $\in \mathbb{F}_2^{n_m}$:
 	
 	\begin{enumerate}
 			\item \textbf{Error-detectability and self correction}
 		\begin{equation}\label{eq:detectorspace}
 			\Pi_{\Acurvy \cap \Bcurvy} U \Pi_{\Acurvy}=U\Pi_{\Acurvy}
 		\end{equation}
 		\begin{equation}\label{eq:selfcorrectiondef}
 			\Pi_{\Bcurvybar}(\vec{m}_b) a U \Pi_{\Acurvy}=e^{i\alpha_b} 	\Pi_{\Bcurvybar}(\vec{m}_b)  U \Pi_{\Acurvy}\hspace{1em} \forall\, a \in \Acurvybar,
 		\end{equation}
 		\item \textbf{Logical preservation}
 			\begin{equation}\label{eq:logicalpresdef}
 		 \Pi_{\Bcurvybar}(\vec{m}_b) U \Pi_{\Acurvy} \propto   U_{\Acurvy,\Bcurvy}(\vecmb) \Pi_{\Acurvy} \, \forall \vecmb,
 		\end{equation}
 		where $ U_{\Acurvy,\Bcurvy}(\vecmb): \Hil_{\Acurvy}\rightarrow \Hil_{\Bcurvy}(\vecmb)$ is unitary.
 		\item \textbf{Logical equivalence}.  
 		\begin{equation}\label{eq:logicalequivdef}
 			\Pi_{\Acurvybar}\Pi_{\Bcurvybar}(\vec{m}_b) U \Pi_{\Acurvy} = e^{i (\phi_b-\phi'_b)}	\Pi_{\Acurvybar}\Pi_{\Bcurvybar}(\vec{m}'_b)U \Pi_{\Acurvy}.
 		\end{equation}
 	\end{enumerate}
 	 Here, $e^{i\alpha_b}$, $e^{i\phi_b}$ and $e^{i\phi'_b}$ in Eqn.~\eqref{eq:selfcorrectiondef} and Eqn.~\eqref{eq:logicalequivdef} are phases dependent on $\vecmb$ and $\vecmb$, $\vecmbp$ respectively.
 \end{definition}

 
 
 \subsection{Error detect-ability and self correcting errors}\label{sec:errordetect}
 We begin our discussion of definition \ref{def:genunitaries} by showing how one arrives at conditions~\eqref{eq:detectorspace} and~\eqref{eq:selfcorrectiondef}. We first show that it is necessary to preserve the intersections of two subsequent ISGs in order to be able to correct errors. Then, we additionally require that certain errors must be self-corrected. 
 We do not aim here at providing a full overview on error detection and correction, nor do we aim at stating necessary and sufficient conditions for error detectability in Floquet codes, as the actual error detection and correction procedure is tightly linked to the chosen measurement sequences (which are not always fully captured by the specification of the ISGs). See~\cite{davydovaQuantumComputationDynamic2024,fuErrorCorrectionDynamical2025,rodatzFloquetifyingStabiliserCodes2024a} for a more detailed analysis of error correction in Floquet codes.
 \subsubsection{Error detectability}
We begin by arguing that if one requires a sequence $\Acurvy \rightarrow \Bcurvy$ of two conjugate ISGs to be able to detect errors, then all information about non-trivial errors must be contained within their intersection $\Acurvy \cap \Bcurvy$.

Take for example a reversible pair $\Acurvy \leftrightarrow \Bcurvy$ where each stabiliser group defines a QEC. At time step $t=1$, consider the error $b \in \Bcurvybar$ acting on $\ket{\Psi} \in \Hil_{\Acurvy}$. This error is subsequently absorbed into the measurement of the conjugate basis elements:
\begin{equation}
	\Pi_{\Bcurvybar} b \ket{\Psi} = \pm \Pi_{\Bcurvybar} \ket{\Psi},
\end{equation}
where the sign depends on the measurement result for this particular $b$, but it merely constitutes a global phase. We see that an error generated by $\Bcurvybar$ does not change any logical information. As it commutes with $\Acurvy \cap \Bcurvy$ element-wise by definition and does not change the probability distribution of the measurement results $\vec{m}_b$, it also is \emph{undetectable}. We will see that if we want to be able to correct errors in a dynamical code, these must stay the only undetectable Pauli errors. 
To be precise, we consider an error $E$ at time step $t$ in a dynamical code to be \emph{undetectable} if its syndrome stabilisers are fully contained within $\Acurvy_{t}\setminus (\Acurvy_t \cap \Acurvy_{t+1})$. After measuring $\Acurvy_{t+1} \setminus (\Acurvy_t \cap \Acurvy_{t+1})$ one fully removes all information about the signs of  $\Acurvy_t \setminus (\Acurvy_t \cap \Acurvy_{t+1})$. Furthermore, we say that a Pauli $P$ has trivial logical effect on a stabiliser code $\Acurvy$ if it acts as the identity on $\mathcal{N}(\Acurvy)/\Acurvy$.

 \begin{proposition}\label{cor:undetect}
 	Let $\Acurvy \rightarrow \Bcurvy$ be a Floquet transition between the time steps $t=0$ and $t=1$ induced by measuring the conjugate basis elements $\{b_j\} \in \Bcurvy \setminus (\Acurvy \cap \Bcurvy)$. 
 	The entire information needed to correct errors must be contained within $\Acurvy \cap \Bcurvy$. That is, if there exists an error $E$ at time step $t=0$ that anti-commutes with at least one element of $\Acurvy \setminus (\Acurvy \cap \Bcurvy)$ but commutes with all of $\Acurvy \cap \Bcurvy$ element-wise and that is \emph{not} a product of elements in $\Acurvy$ and $\Bcurvybar$, then $E$ becomes an undetectable error with non-trivial logical effect at time step $t=1$.
 \end{proposition}
 \begin{proof}
 	Let $E$ be a Pauli that anti-commutes with an $a \in \Acurvybar$, but commutes with all $s \in \Acurvy \cap \Bcurvy$, i.e. it is undetectable at $t=0$. There are two cases. Either $E$ commutes with all of $\Bcurvybar$ element-wise or there exists at least one $b \in \Bcurvybar$ with which $P$ anti-commutes with. In both cases, due to proposition~\ref{prop:Psurvives} one can find a low weight $E'=Ea$, for some $a \in \Acurvybar$, that commutes with all of $\Bcurvybar$ and thus with all of $\Bcurvy$ element-wise. This means $E'$ is now an element of $\mathcal{N}(\Bcurvy)$. In order for $E'$ to have trivial logical effect, it must be an element of $\Bcurvy$, which means $E$ must have been a product of elements in $\Acurvybar$ and $\Bcurvy$, or, equivalently, a product of elements in $\Acurvy$ and $\Bcurvybar$. Thus, if $E$ is not such a product, it cannot have trivial logical effect, meaning $E'$ is an element of $\mathcal{N}(\Bcurvy)\setminus \Bcurvy$.
  
 \end{proof}
 As error-correction is handled by $\Acurvy \cap \Bcurvy$, our first condition for our generalised logical unitary $U$ is \textbf{error detectability} (Eqn.~\eqref{eq:detectorspace} in definition~\ref{def:genunitaries}):
 \begin{equation}
 \Pi_{\Acurvy \cap \Bcurvy}	U\Pi_{\Acurvy}=U\Pi_{\Acurvy}.
 \end{equation}
%
%
  \subsubsection{Self-correction of errors}\label{sec:selfcorrecterror}
  
  After leaving the codespace $\Hil_{\Acurvy}$  (under the action of $U$) but remaining in $\Hil_{\Acurvy \cap \Bcurvy}$, elements in $\Acurvybar$ can become errors. As we subsequently measure $\Bcurvybar$ and fully randomise the signs of $a \in \Acurvybar$, these errors can no longer be detected. We therefore require them to self-correct, similar to how errors that are elements of  $\Bcurvybar$ occurring right  before the measurement self-correct in the standard Floquet transition. This means we require for all $a \in \Acurvybar$: 
  
  \begin{equation}\label{eq:selfcorrection}
  	\Pi_{\Bcurvybar} a U \Pi_{\Acurvy} = e^{i\alpha_b} \Pi_{\Bcurvybar} U \Pi_{\Acurvy},
  \end{equation}
  and $e^{i\alpha_b}$ is a phase generally dependent on $\vecmb$.
  This is Eqn.~\eqref{eq:selfcorrectiondef} in definition~\ref{def:genunitaries}.
  The condition of \textbf{self-correction} in turn leads to the following: For any $\vec{m}_b, \vec{m'}_b \in \mathbb{F}_2^{n_m}$ and any $\ket{\Phi}=U\ket{\Psi}_{\Acurvy}$:

  \begin{align}
   \bra{\Phi} \Pi_{\Bcurvybar}(\vec{m}'_b) \ket{\Phi} &= \bra{\Phi} a\Pi_{\Bcurvybar}(\vec{m}_b)a \ket{\Phi} \\
   &=\bra{\Phi} e^{-i\alpha_b}\Pi_{\Bcurvybar}(\vec{m}_b)e^{i\alpha_b}\ket{\Phi}\\
   &=\bra{\Phi} \Pi_{\Bcurvybar}(\vec{m}_b) \ket{\Phi} \label{eq:probsame},
  \end{align}
  where $a$ is some $a \in \Acurvybar$. Thus, the self-correction condition also demands that the  $\vec{m}_b$ are measured with uniform probability, i.e.:
  \begin{equation}\label{eq:genUequalprob}
  	\bra{\Psi} U^{\dag} \Pi_{\Bcurvybar}(\vecmb)U\ket{\Psi} =\frac{1}{2^{n_m}},
  \end{equation}
  for all $\ket{\Psi}\!\in \! \Hil_{\Acurvy}$ and $ \vecmb \! \in \! \mathbb{F}_2^{n_m}$.
  This is a sensible requirement to make for a Floquet even outside of the context of logical operators, as the probability to measure $\vec{m}_b$ influences the form of the effective error channel at a later time.  
  
  \subsection{Preservation of logical information and logical equivalence}\label{sec:logicalpres}
  
  Arguably, the main characteristic of Floquet codes are projective measurements, that - while not necessarily being logically trivial - preserve logical information. As already explained in section~\ref{sec:presinfo}, the actual results of the conjugate basis measurements are irrelevant for the encoded logical information. In this section require a generalised logical unitary that leaves the codespace $\Acurvy_t$ to still be compatible with this fundamental principle. That will lead to \eqref{eq:logicalequivdef}.
  
  For a reversible pair $\Acurvy \leftrightarrow \Bcurvy$, let $U$ be a unitary operation on the entire Hilbert space $\Hil^{\otimes n}$. It is followed by the projection $\Pi_{\Bcurvybar}$, where we did not specify the measurement results. 
  As discussed in section~\ref{sec:presinfo}, a fundamental principle of Floquet codes is the preservation of logical information. Thus, in order to preserve logical information, we require the combined action of $\Pi_{\Bcurvybar}$ and $U$ to be proportional to be a purity preserving channel, i.e. a unitary $U_{\Acurvy,\Bcurvy}$: 
  \begin{equation}\label{eq:propU}
  		\Pi_{\Bcurvybar}(\vec{m}_b)U\Pi_{\Acurvy}\propto U_{\Acurvy,\Bcurvy}(\vec{m}_b)\Pi_{\Acurvy}.
  \end{equation}  
 Note that $U_{\Acurvy,\Bcurvy}(\vec{m}_b)$ does not act on the entire space $\Hil^{\otimes n}$, but maps an element of $\Hil_{\Acurvy}$ to an element of $\Hil_{\Bcurvy}(\vec{m}_b)$. We summarise condition~\eqref{eq:propU} as the principle of \textbf{logical preservation} (Eqn.~\eqref{eq:logicalpresdef} in definition~\ref{def:genunitaries}). 
  
  Condition~\ref{eq:propU} is very general and the logical evolution of a state $\ket{\Psi} \in \Hil_{\Acurvy}$ can in principle be dependent on the measurement result $\vec{m}_b$. We believe that for the actual practical implementation of dynamical codes it would be inconvenient and intractable to keep up with the measurement results to compute the correct logical operation. Thus, we will further require that the resulting logical operation $U_{\Acurvy,\Bcurvy}$ is the same for all measurement results. To precisely state what that means, we briefly introduce the concept of logically equivalent paths. \\
 We would like to express that the state $\Pi_{\Bcurvybar}(\vec{m}_b) \ket{\Psi}$ is logically equivalent to $\Pi_{\Bcurvybar}(\vec{m}_b)\ket{\Psi}$ for some $\Psi \in \Hil_{\Acurvy}$. For this note that, as $\{a_i\}$ and $\{b_i\}$ are conjugate bases, for any $\vec{m}'_b \neq \vec{m}_b$, there exists some $a \in \Acurvybar$, such that:
  
  \begin{equation}
  	a \Pi_{\Bcurvybar}(\vec{m}_b) a = \Pi_{\Bcurvybar}(\vec{m}'_b),
  \end{equation}
  which means that:
  \begin{equation}\label{eq:logequiva}
  	\Pi_{\Bcurvybar}(\vec{m}'_b) \Pi_{\Acurvy} = a 	\Pi_{\Bcurvybar}(\vec{m}_b) \Pi_{\Acurvy}.
  \end{equation}
  We thus introduce the following definition for \emph{logically equivalent states}:
 \begin{definition}\label{def:logicalequiv}
  	Let $\Acurvy \leftrightarrow \Bcurvy$ be a reversible pair of stabiliser groups. We say that a state $\ket{\Psi} \in \Hil_{\Bcurvy}(\vec{m}_b)$ is \emph{logically equivalent} to a $\ket{\Phi} \in \Hil_{\Bcurvy}(\vec{m}'_b)$ if
  	\begin{equation}
  		\Pi_{\Acurvy}\ket{\Psi}=	e^{i\phi}\Pi_{\Acurvy}\ket{\Phi},
  	\end{equation}
  	where $e^{i\phi}$ is some phase with $\phi \in [0,2\pi]$.
  \end{definition}
  Per definition~\ref{def:logicalequiv} and due to Eqn.~\eqref{eq:logequiva} $\Pi_{\Bcurvybar}(\vec{m}_b) \ket{\Psi}$ and $\Pi_{\Bcurvybar}(\vec{m}'_b)\ket{\Psi}$ are then logically equivalent for some $\ket{\Psi} \in \Hil_{\Acurvy}$. With this, we can introduce the condition of  
 \textbf{logical equivalence.} For all $\vec{m}_b, \vec{m'}_b \in \mathbb{F}_2^{m_n}$, we require:
\begin{equation}
  \Pi_{\Acurvy}\Pi_{\Bcurvybar}(\vec{m}_b)U\Pi_{\Acurvy}\propto e^{i\phi_{b,b'}}\Pi_{\Acurvy}\Pi_{\Bcurvybar}(\vec{m}'_b)U\Pi_{\Acurvy}, \label{eq:logicalequiv}
 \end{equation}
where $e^{i \phi_{b,b'}}$ is a phase dependent on $\vecmb$ and $\vecmb'$. \newline 
To obtain the exact proportionality constant and summarise logical equivalence and logical preservation in Eqn.~\eqref{eq:logicalequivdef}, we show the following lemma:
  \begin{lemma}\label{lem:Uasame}
 	Let $\Acurvy \leftrightarrow \Bcurvy$ be a reversible pair of stabiliser groups and let $U$ be a unitary fulfilling conditions~\eqref{eq:detectorspace},~\eqref{eq:selfcorrectiondef},~\eqref{eq:propU} and Eqn.~\eqref{eq:logicalequiv}. Then, for all $\vecmb \in \mathbb{F}_2^{n_m}$ and a fixed $\Acurvybar$:
 	\begin{equation}
 	\Pi_{\Acurvy}	\Pi_{\Bcurvybar}(\vecmb)U\Pi_{\Acurvy}=\frac{1}{2^{n_m}}e^{i\phi_b} U_{\Acurvy} \Pi_{\Acurvy} \label{eq:PaPbUprop},
 	\end{equation}
 	where $U_{\Acurvy}: \Hil_{\Acurvy} \rightarrow \Hil_{\Acurvy}$ is the same for all $\vecmb$ and $e^{i\phi_b}$ is a phase dependent on $\vecmb$. Moreover:
 	\begin{equation}
 		\Pi_{\Bcurvybar}(\vecmb)U\Pi_{\Acurvy}=\frac{1}{2^{(n_m/2)}} e^{i\phi_b} K_{\Acurvy,\Bcurvy}(\vecmb)U_{\Acurvy}\Pi_{\Acurvy}
 	\end{equation}
 	for all $\vecmb \in \mathbb{F}_2^{n_m}$. $U_{\Acurvy}$ is a unitary that preserves $\Hil_{\Acurvy}$ and is independent of $\vecmb$ and $e^{i\phi_b}$ is the same phase as in Eqn.~\eqref{eq:PaPbUprop}.
 \end{lemma}
 \begin{proof}
 	Take any two $\vecmb$, $\vecmbp$. Then, Eqn.~\eqref{eq:propU} implies:
 	\begin{equation}
 	\Pi_{\Acurvy}	\Pi_{\Bcurvybar}(\vecmb) U \Pi_{\Acurvy} \propto \Pi_{\Acurvy} U_{\Acurvy,\Bcurvybar}(\vecmb) \Pi_{\Acurvy}.
 	\end{equation}
 	As $U_{\Acurvy,\Bcurvybar}(\vecmb): \Hil_{\Acurvy} \rightarrow \Hil_{\Bcurvy}$ is bijective by definition, we can always find a $U_A (\vecmb) : \Hil_{\Acurvy} \rightarrow \Hil_{\Acurvy}$ such that:
 	\begin{equation}
 		U_{\Acurvy,\Bcurvy}(\vecmb) = K_{\Acurvy,\Bcurvy} U_{\Acurvy}(\vecmb),
 	\end{equation}
 	where $K_{\Acurvy,\Bcurvy}$ is the Floquet transition operator from definition~\ref{def:KdefHH}.
 	Thus, we obtain:
 	\begin{align}
 		\Pi_{\Acurvy}	\Pi_{\Bcurvybar}(\vecmb) U \Pi_{\Acurvy} &\propto \Pi_{\Acurvy}  K_{\Acurvy,\Bcurvy} U_{\Acurvy}(\vecmb) \Pi_{\Acurvy}\\
 		 &\propto U_{\Acurvy}(\vecmb) \Pi_{\Acurvy},\label{eq:propUunspec}
 	\end{align}
 	where we utilised Eqn.~\eqref{eq:papbpaid} and inserted the definition of $K_{\Acurvy,\Bcurvy}$. Due to condition~\eqref{eq:logicalequiv}, we conclude that:
 \begin{equation}\label{eq:Uasame}
 	U_{\Acurvy}(\vecmb)=U_{\Acurvy}(\vecmbp):= U_{\Acurvy},
 \end{equation}
 for all $\vecmb, \, \vecmbp \in \mathbb{F}_2^{n_m}$.
Therefore we have
 \begin{equation}\label{eq:actionsame}
 	\Pi_{\Acurvybar}\Pi_{\Bcurvybar}(\vecmb)U\Pi_{\Acurvy}\propto 	\Pi_{\Acurvybar}\Pi_{\Bcurvybar}(\vecmbp)U\Pi_{\Acurvy}
 \end{equation}
 and:
  \begin{equation}\label{eq:actionsameKab}
  \Pi_{\Bcurvybar}(\vecmb)U\Pi_{\Acurvy}\propto 	 K_{\Acurvy,\Bcurvy}(\vecmb)U_{\Acurvy}\Pi_{\Acurvy}
 \end{equation}
  for all $\vecmb$, $\vecmbp$ $\in \mathbb{F}_2^{n_m}$. \newline
It remains to show the exact proportionality constants in Eqn.~\eqref{eq:propUunspec} and Eqn.~\eqref{eq:actionsameKab}. Let $c_b \in \mathbb{C}$ denote this constant in Eqn.~\eqref{eq:actionsameKab}, i.e.:
\begin{equation}\label{eq:PiBUPiAconst}
		\Pi_{\Bcurvybar}(\vecmb)U \Pi_{\Acurvy}=c_b K_{\Acurvy,\Bcurvy}(\vecmb)U_{\Acurvy}\Pi_{\Acurvy}.
\end{equation}
In the remainder of the proof we will leave out the dependency of $\Pi_{\Bcurvybar}(\vecmb)$ on $\vecmb$ and simply write $\Pi_{\Bcurvybar}$. \newline 
As discussed in section~\ref{sec:errordetect}, error detectability and self correction, i.e. conditions~\eqref{eq:detectorspace} and~\eqref{eq:selfcorrectiondef}, imply that for any $\ket{\Psi} \in \Hil_{\Acurvy}$ the probability to measure $\vecmb$ is uniform (See Eqn.~\eqref{eq:genUequalprob}.). Thus:
\begin{align}
	\bra{\Psi} \! U^{\dag}	\Pi_{\Bcurvybar} U  \ket{\Psi}&=\! 2^{n_m}\!	\bra{\Psi}\! U_{\Acurvy}^{\dag} K^{\dag}_{\Acurvy,\Bcurvy}c^*_b \! \cdot \! c_b K_{\Acurvy,\Bcurvy} U_{\Acurvy} \!\ket{\Psi} \\
	&=|c_b|^2 \bra{\Psi}    \ket{\Psi} \\
	&=\frac{1}{2^{n_m}}
\end{align}
where we inserted Eqn.~\eqref{eq:PiBUPiAconst} in the first and Eqn.~\eqref{eq:genUequalprob} in the last line. We conclude that $c_b$ in Eqn.~\eqref{eq:PiBUPiAconst} is $c_b = 2^{-(n_m/2)}e^{i\phi_b}$, where $\phi_b$ generally depends on $\vecmb$. \newline 
Now we multiply the left and right hand side of Eqn.~\eqref{eq:PiBUPiAconst} by $\Pi_{\Acurvy}$ and insert $c_b = 2^{-(n_m/2)}e^{i\phi_b}$  to obtain (We leave out the dependency of $\Pi_{\Bcurvybar}$ on $\vecmb$): 
 \begin{align}
 \Pi_{\Acurvy}\Pi_{\Bcurvybar}U \Pi_{\Acurvy}&=c_b \Pi_{\Acurvy}K_{\Acurvy,\Bcurvy}U_{\Acurvy}\Pi_{\Acurvy}\\
 &=\frac{2^{(n_m/2)}e^{i\phi_b}}{2^{(n_m/2)}}\Pi_{\Acurvy} \Pi_{\Bcurvy}\Pi_{\Acurvy}U_{\Acurvy}\Pi_{\Acurvy}\label{eq:papbupaconst2}\\
 &=\frac{e^{i\phi_b}}{2^{n_m}}\Pi_{\Acurvy}  U_{\Acurvy}\Pi_{\Acurvy} \label{eq:papbupaconst3}\\
 &=\frac{e^{i\phi_b}}{2^{n_m}}  U_{\Acurvy}\Pi_{\Acurvy},
  \end{align}	
  whereby we inserted the definition of $K_{\Acurvy,\Bcurvy}$ (see definition~\ref{def:KdefHH}) in Eqn.~\eqref{eq:papbupaconst2} and utilised proposition~\ref{prop:pbpapbidHH} in Eqn.~\eqref{eq:papbupaconst3}.
 	\end{proof}
Lemma~\ref{lem:Uasame} lets us insert the exact proportionality constant into Eqn.~\eqref{eq:logicalequiv} and derive Eqn.~\eqref{eq:logicalequivdef} in definition~\ref{def:genunitaries}:
\begin{align}
	\Pi_{\Acurvy}\Pi_{\Bcurvybar}(\vecmb)U\Pi_{\Acurvy}&=\frac{1}{2^{n_m}}e^{i\phi_b}U_{\Acurvy}\Pi_{\Acurvy}\\
	&= e^{i(\phi_b-\phi'_b)}\frac{1}{2^{n_m}} e^{i\phi'_b}U_{\Acurvy}\Pi_{\Acurvy}\\
	&=e^{i(\phi_b-\phi'_b)} \Pi_{\Acurvy}\Pi_{\Bcurvybar}(\vecmbp)U\Pi_{\Acurvy},
\end{align}
where we used Eqn.~\eqref{eq:PaPbUprop} in the first and last line and inserted identity in the second.

\section{Canonical form of generalised logical unitaries in Floquet transitions}\label{sec:canonicalform}

In this section we derive a canonical form for generalised logical unitaries:
 
 \begin{theorem}\label{thm:canonicalform}
 	Let $\Bcurvy_e \subseteq \Bcurvybar$ and $\{\phi_b\}$ be a set of angles $\phi_b \in [0,2 \pi)$. For any generalised logical unitary as defined in definition~\ref{def:genunitaries}, one can find the following representation:
 		
 		\begin{equation}
 			U \Pi_{\Acurvy} = \exp(\sum_{b \in \Bcurvy_e} i\phi_b b) U_{\Acurvy} \Pi_{\Acurvy} \label{eq:UequalexpUa}
 		\end{equation}
 		where $U_{\Acurvy}\Pi_{\Acurvy}=\Pi_{\Acurvy}U_{\Acurvy}=\Pi_{\Acurvy}U_{\Acurvy}\Pi_{\Acurvy}$, and $U_{\Acurvy}$ acts as a unitary on $\Hil_{\Acurvy}$. 
 	\end{theorem}
 	
 	\begin{proof}
 		We start by decomposing $U$ into the Pauli basis:
 		\begin{equation}
 			U = \sum_P \alpha_P P.
 		\end{equation}
 		We abbreviate $\Acurvy \cap \Bcurvy$ by $\Scurvy$.
 		Error detectability (Eqn.~\eqref{eq:detectorspace}) implies:
 		\begin{align}
 		U\Pi_{\Acurvy}	&=\left(\sum_{P \in \mathcal{N}(\Scurvy)} \alpha_P P +\sum_{P \notin \mathcal{N}(\Scurvy)} \alpha_P P \right) \Pi_{\Acurvy}\\
 	&=	\Pi_{\Scurvy}\left(\sum_{P \in \mathcal{N}(\Scurvy)} \alpha_P P +\sum_{P \notin \mathcal{N}(\Scurvy)} \alpha_P P \right) \Pi_{\Acurvy}. \label{eq:splitupU}
 	\end{align}
 	Note that for any $P \notin \mathcal{N}(\Scurvy)$ we have $\Pi_{\Scurvy}P \Pi_{\Acurvy}=0$. Therefore, the second sum in Eqn.~\eqref{eq:splitupU} is zero. Moreover, we can commute $\Pi_{\Scurvy}$ through the sum over elements in $\mathcal{N}(\Scurvy)$ and re-absorb it into $\Pi_{\Acurvy}=\Pi_{\Scurvy}\Pi_{\Acurvybar}$. We arrive at:
 	\begin{align}
 		U\Pi_{\Acurvy}	&=\Pi_{\Scurvy}\left(\sum_{P \in \mathcal{N}(\Scurvy)} \alpha_P P\right) \Pi_{\Acurvy}\\
 		&= \left(\sum_{P \in \mathcal{N}(\Scurvy)} \alpha_P P\right) \Pi_{\Acurvy}.\label{eq:sumovernormS}
 		\end{align}
 	 We have obtained a representative for $U$ that commutes with $\Pi_{\Scurvy}$. \newline
 	We will now show that $U\Pi_{\Acurvy}$ must factorise into a product of an operator that preserves $\Hil_{\Acurvy}$ and an operator that is spanned by the elements of $\Bcurvybar$. According to proposition~\ref{prop:Psurvives}, we can find a $b \in \Bcurvybar$ for each $P$ in the sum in Eqn.~\eqref{eq:sumovernormS} such that $Q:=b P \in \mathcal{N}(\Acurvy)$. We thus rewrite the sum as:
 	\begin{align}
 		U\Pi_{\Acurvy}&=\sum_{Q\in \mathcal{N}(\Acurvy)} \left(\sum_{b \in \Bcurvybar} \alpha_{b Q} b\right)Q\Pi_{\Acurvy}.\label{eq:splitBequiv}
 	\end{align}
 	The summands in Eqn.~\eqref{eq:splitBequiv} generally contain multiple $Q$ with the same action on $\Pi_{\Acurvy}$, i.e. there can be a $Q'$ and a $Q$, such that $Q \Pi_{\Acurvy}= Q'\Pi_{\Acurvy}$. More precisely, the logical action of elements of $\mathcal{N}(\Acurvy)$ is determined by the quotient group $\sigma:=\mathcal{N}(\Acurvy)/\Acurvy$ and when $Q$ and $Q'$ are both representatives of the same $q \in \sigma$, then $Q'\Pi_{\Acurvy}=Q\Pi_{\Acurvy}$. 
 	For each $q \in \sigma$, we define:
 	\begin{equation}
 		L_q:= \sum_{Q \in q} \left(\sum_{b \in \Bcurvybar} \alpha_{bQ} b\right).
 	\end{equation}
 	Then, $U\Pi_{\Acurvy}$ becomes:
 	\begin{align}\label{eq:UequalLqQ}
 	U\Pi_{\Acurvy}&=  \sum_{q \in \sigma} L_q Q\Pi_{\Acurvy}, 
 	\end{align}
 	where $Q$ is any representative of $q$. Importantly, we have chosen a $Q \in q$ for each equivalence class $q \in \sigma$ and reduce each $Q' \Pi_{\Acurvy}$ to $Q \Pi_{\Acurvy}$ for every other $Q' \in q$. We will now show that due to the principle of logical equivalence, i.e. condition~\eqref{eq:logicalequivdef}, $U \Pi_{\Acurvy}$ in Eqn.~\eqref{eq:UequalLqQ} splits up into the product of an operator $L_A$ that preserves $\Hil_{\Acurvy}$ and an operator $L_b$ that preserves $\Hil_{\Bcurvy}(\vecmb)$. To do so, we first compute $\Pi_{\Bcurvybar}(\vecmb)U\Pi_{\Acurvy}$. It collapses each $L_Q$ onto a scalar $l_Q(\vecmb)$ dependent on the measurement results $\vecmb$ (which determine the eigenvalues of $\Bcurvybar$):
 	 \begin{align}
 	 	\Pi_{\Bcurvybar}(\vecmb)U\Pi_{\Acurvy}&=\Pi_{\Bcurvybar}(\vecmb)\left( \sum_{q \in \sigma} L_q Q\right)\Pi_{\Acurvy}\\
 	 	&=\Pi_{\Bcurvybar}(\vecmb)\left( \sum_{q \in \sigma} l_Q(\vecmb) Q\right)\Pi_{\Acurvy}.\label{eq:UdecomplqQ}
 	 \end{align}
 	 Condition~\eqref{eq:logicalequivdef} implies that for any two $\vecmb$, $\vecmbp$:
 	 \begin{align}
   \Pi_{\Acurvy}	\Pi_{\Bcurvybar}(\vecmb)U\Pi_{\Acurvy}
 	& =e^{i(\phi_{b}-\phi'_b)} \Pi_{\Acurvy}\Pi_{\Bcurvybar}(\vecmbp) U \!\Pi_{\Acurvy}.\label{eq:Usameephiphiprime}
 	 	 \end{align} 
 	 Each $Q$ individually commutes with $\Pi_{\Acurvy}$ and thus the entire sum in Eqn.~\eqref{eq:UdecomplqQ} commutes with $\Pi_{\Acurvy}$. Therefore, we can pull $\Pi_{\Acurvy}$ through and use Eqn.~\eqref{eq:Usameephiphiprime} and proposition~\ref{prop:pbpapbidHH} to conclude:
 	 	 \begin{align}\label{eq:operatorssame}
 	  \left( \sum_{q \in \sigma} l_Q(\vecmb) Q\right)  \Pi_{\Acurvy}	  =e^{i(\phi_b-\phi'_b)} \left( \sum_{q \in \sigma} l_Q(\vecmbp) Q\right).
 	 \end{align}
 	 The terms $Q\Pi_{\Acurvy}$ in the sum in Eqn.~\eqref{eq:operatorssame} are linearly independent, as there is only one term per equivalence class.
 	  Thus, Eqn.~\eqref{eq:operatorssame} implies that for all $q \in \sigma$ (and thus for the respective representative $Q$):
 	 \begin{equation}\label{eq:lblbprimeprop}
 	 	\frac{l_Q(\vecmb)}{l_Q(\vecmbp)}= e^{i(\phi_b-\phi'_b)}
 	 \end{equation}
 	 where $e^{i(\phi_b-\phi'_b)}$ is a phase independent of $Q$. Eqn.~\eqref{eq:lblbprimeprop} will serve as an important tool to show how $U\Pi_{\Acurvy}$ splits up. We will now insert identity as a decomposition into the projection operators $\Pi_{\Bcurvybar}(\vecmb)$, i.e.:
 	 \begin{equation}\label{eq:identitydecomp}
 	 	\mathbb{1} =\sum_{\vecmb}\Pi_{\Bcurvybar}(\vecmb).
 	 \end{equation}
 	 For notational convenience, we will denote the all positive measurement result $\vec{m}_b=(1,1,\ldots,1)$ by $\vec{m}$. The projection operator onto the joint \emph{positive} eigenspace of $\Bcurvybar$ is then given by $\Pi_{\Bcurvybar}(\vec{m})$. Then, we can write:
 	 \begin{equation}
 	 	\mathbb{1} =\Pi_{\Bcurvybar}(\vec{m})+\sum_{\vecmb \neq \vec{m}}\Pi_{\Bcurvybar}(\vecmb).
 	 \end{equation}
 	 For $\vec{m}$ and any $\vecmb \neq \vec{m}$, Eqn.~\eqref{eq:lblbprimeprop} holds, i.e.:
 	 	 \begin{equation}\label{eq:lblbprimeprop2}
 	 	\frac{l_P(\vecmb)}{l_P(\vec{m})}= e^{i(\phi_{b}-\phi)},
 	 \end{equation}
 	 where we simplified the notation for the phase $e^{i\phi}$ induced by the measurement result $\vec{m}$.
 	 We then use Eqn.~\eqref{eq:lblbprimeprop2} and insert identity as defined in Eqn.~\eqref{eq:identitydecomp} to derive the following form of $U \Pi_{\Acurvy}$:
 	\begin{align}
 		U \Pi_{\Acurvy} &= \left(\Pi_{\Bcurvybar}(\vec{m})+\sum_{\vecmb \neq \vec{m}}\Pi_{\Bcurvybar}(\vecmb)\right) U \Pi_{\Acurvy} \\
 			&= \Pi_{\Bcurvybar}(\vec{m}) \left(\sum_{q} l_Q(\vec{m}) Q \right)\Pi_{\Acurvy}\label{eq:insertPbsum} \\
 			 &\hspace{2.5em} + \!\sum_{\vecmb \neq \vec{m}}\! \Pi_{\Bcurvybar}(\vecmb) e^{i(\phi_{b}-\phi)} \left(\sum_{q} l_Q(\vec{m}) Q \right)\!\Pi_{\Acurvy}.\notag 
 			\end{align}
 	In line~\eqref{eq:insertPbsum}, we split up the set $\{\vecmb\}$ into the fixed $\vec{m}$ and all other $\vecmb \neq \vec{m}$. The second sum only runs over all $\vecmb \neq \vec{m}$. In the following, we will write $l_Q$ for $l_Q(\vec{m})$ and define:
 	\begin{equation}
 	L_A :=	\sum_{q} l_q Q. 
 	\end{equation}
 	Simplifying line~\eqref{eq:insertPbsum} leads to:		
 	\begin{align}
 			 U \Pi_{\Acurvy} &\!=\! \!\left[\Pi_{\Bcurvybar}(\vec{m})\!+\!e^{-i\phi}  \!\sum_{\vecmb \neq \vec{m}} e^{i\phi_{b}} \Pi_{\Bcurvybar}(\vecmb) \!\right]\! L_A \Pi_{\Acurvy} \label{eq:projdecomp}\\
 			&=L_b L_A\Pi_{\Acurvy},\label{eq:UPiALaLbPiA}
 			\end{align}
 			where we defined:
 		\begin{equation}
 			L_b := \Pi_{\Bcurvybar}(\vec{m})+ e^{-i\phi}\sum_{\vecmb \neq \vec{m}} e^{i\phi_b} \Pi_{\Bcurvybar}(\vecmbp). \label{eq:Lbdef}
 		\end{equation}  
 	 $L_b$, as defined in Eqn.~\eqref{eq:Lbdef} only acts as a phase on any $\ket{\Psi} \in \Hil_{\Bcurvy}(\vecmb)$ for all $\vecmb \in \mathbb{F}_2^{n_m}$ and thus preserves its norm and the orthogonality of two orthogonal vectors. As the subspaces $\Hil_{\Bcurvy}(\vecmb)$ span the entire Hilbert space, we can conclude that $L_b$ is unitary on the entire Hilbert space. For notational convenience, we will re-define:
 	 \begin{equation}
 	 	U_b := e^{i\phi}L_b,
 	 \end{equation}
 	 or written out:
 	 \begin{equation}
 	 	U_b:=  \sum_{\vecmb} e^{i\phi_b} \Pi_{\Bcurvybar}(\vecmbp), \label{eq:Lbdef2}
 	 \end{equation}
 	 where the phase for $\vecmb = \vec{m}$ is the  $e^{i\phi}$ from Eqn.~\eqref{eq:Lbdef}.
 	It remains to show that $L_A$ acts as a unitary on $\Hil_{\Acurvy}$. Lemma~\ref{lem:Uasame} requires that for any $\vecmb$:
 		\begin{align}
 	 \Pi_{\Acurvy}	\Pi_{\Bcurvybar}(\vec{m}_b)U\Pi_{\Acurvy}&=\frac{1}{2^{n_m}} e^{i\phi_b}U_{\Acurvy}\Pi_{\Acurvy},
 	\end{align}
 	 where $U_{\Acurvy}:\Hil_{\Acurvy} \rightarrow \Hil_{\Acurvy}$ is unitary on $\Hil_{\Acurvy}$.
 	Therefore, if we insert $U\Pi_{\Acurvy}=L_b L_A\Pi_{\Acurvy}$ and Eqn.~\eqref{eq:Lbdef2}, we obtain:
 	\begin{align}
 	\Pi_{\Acurvy}\Pi_{\Bcurvybar}(\vec{m}_b)L_b L_A\Pi_{\Acurvy}&=	\Pi_{\Acurvy}\Pi_{\Bcurvybar}(\vec{m}_b)e^{i\phi_b} L_A\Pi_{\Acurvy}\\
 	&=\Pi_{\Acurvy}\Pi_{\Bcurvybar}(\vec{m}_b)\Pi_{\Acurvy}e^{i\phi_b} L_A\Pi_{\Acurvy}\label{eq:movePiApastL_A}\\
 	&=\frac{1}{2^{n_m}}    e^{i\phi_b}L_A \Pi_{\Acurvy} \label{eq:insertprop}\\
 	& =\frac{1}{2^{n_m}} e^{i\phi_b}U_{\Acurvy}\Pi_{\Acurvy}.
 		\end{align}
 where we made use of the fact that $L_A$ commutes with $\Pi_{\Acurvy}$ in line~\eqref{eq:movePiApastL_A} and used proposition~\ref{prop:pbpapbidHH} in line~\eqref{eq:insertprop}. We conclude that $L_A$ acts as a unitary on $\Hil_{\Acurvy}$. 
We will now show that each $U_b$, as defined in Eqn.~\eqref{eq:Lbdef2} has a representative of the form:
\begin{equation}
	U_b = \exp(\sum_{b \in \Bcurvybar}i \phi_b b),
\end{equation}
where each $\phi_b \, \in \, [0,2\pi]$.
We already know $U_b$ is block-diagonal.
One can quickly show, as projection operators are idempotent and simplify the matrix exponential, that each $U_b$ can be written as:
\begin{equation}
	U_b = \exp(i\sum_{\vecmb}\phi_b \Pi_{\Bcurvybar}(\vecmb)).
\end{equation}
It remains to show that the elements of $\Bcurvybar$ span the same vector space as the one spanned by $\Pi_{\Bcurvybar}(\vecmb)$. There are $2^{n_m}$ linearly independent elements of $\Pi_{\Bcurvybar}(\vecmb)$. There are $2^{n_m}$ linearly independent elements in $\Bcurvybar$, i.e. $\mathbb{1}$, $b_0$, $b_1$, $b_0b_1$, etc. Each of them is block-diagonal, thus each of them is an element of $\langle \Pi_{\Bcurvybar}(\vecmb)\rangle$, where $\langle . \rangle$ denotes the span over $\mathbb{R}$. Hence, $\Bcurvybar$ is a basis for $\langle \Pi_{\Bcurvybar}(\vecmb)\rangle$, and one can equivalently write:
 \begin{equation}
 	U_b = \exp(i\sum_{b \in \Bcurvybar}\alpha_b b),
 \end{equation}
where $\alpha_b \in \mathbb{R}$. As all $b \in \Bcurvybar$ commute element-wise and each individual $\exp(i\alpha_b b)$ simplifies to $\cos(\phi_b)\mathbb{1}+i \sin(\phi_b)b$, we can choose $\alpha_b$ from $[0,2\pi]$, which finalises the proof.
 \end{proof}

 \section{The Bravyi-K\"onig theorem for logical operations beyond code preservation}\label{sec:BKbeyondcodepres}
 
 In this section we will show the main result of our work, namely that the Bravyi-K\"onig theorem also holds for generalised logical operations as defined in section~\ref{sec:canonicalform}. 
 
 We denote the system size by $n$. Strictly speaking, we are discussing a \emph{family} of circuits $\{U_i\}$ containing the implementation of the corresponding logical unitary $U_L: \Hil^{\otimes k} \rightarrow \Hil^{\otimes k}$, where $k$ is the number of logical qubits, for each system size $n$. We say that the family of circuits (or simply the circuit) is \emph{constant depth} if for each $n$ the range $r$ and depth $h$ of the circuit is constant with respect to $n$, i.e. $r, h = 
 \mathcal{O}(1)$. 
 
  Given two locally conjugate stabiliser groups $\Acurvy$ and $\Bcurvy$, recall the canonical form of a generalised logical operation in a transition $\Acurvy \rightarrow \Bcurvy$ (theorem~\ref{thm:canonicalform}.):
 \begin{equation}\label{eq:Ugenrep}
 	U = \exp(i \sum_{b \in \Bcurvy_{e}} \phi_b b)U_a = e^{i  \vec{\phi} \vec{b} }U_a
 \end{equation} 
 where $U_a$ preserves the codespace $\Hil_{\Acurvy}$, $\phi_b \in [0,2\pi)$ and $\Bcurvy_e \subseteq \Bcurvybar$.
 We remove all $e^{i \phi_b b}$ with $\phi_b \in \{2\pi k, (\pi/2) k \}$, $k \in \mathbb{Z}$, because those constitute transversal Pauli operators. Then, when we say that $e^{i\vec{\phi}\vec{b}}$ does not permit an implementation as a constant-depth circuit, it means that for any constant $C>0$, there exists a system size $n$ for which there exists a $b \in B_e$, such that $\mathrm{diam}(\mathrm{supp}(b)) > C$. In the rest of this section, we will prove theorem~\ref{thm:explocal}:~\begin{theorem}\label{thm:explocal}
	Let $\Acurvy$ and $\Bcurvy$ be topological stabiliser codes that are locally conjugate stabiliser groups with local generating sets $\{a_i\}$ and $\{b_i\}$. Let $\Bcurvy_e \subseteq \Bcurvybar$. Let
	\begin{equation}\label{eq:Uprod}
		U=  e^{i  \vec{\phi} \vec{b} }U_{\Acurvy},
	\end{equation}
	where $U_{\Acurvy}$ preserves the codespace $\mathcal{H}_{\Acurvy}$. If $e^{i\vec{\phi}\vec{b}}$ has no implementation as a constant-depth circuit, then $U$ cannot be constant-depth either.
\end{theorem}
From theorem~\ref{thm:explocal} and the original BK theorem~\ref{thm:BKog}, the main statement of our work immediately follows as a corollary:~ \begin{corollary}\label{thm:main}
 	Let $\Acurvy_1 \rightarrow \Bcurvy_1 \rightarrow \ldots \Acurvy_{\tau}$ denote a sequence $D$-dimensional topological stabiliser codes, where $A_t \subseteq \mathcal{P}_n$ for each $t$ and $\tau =\mathcal{O}(1)$. Let each consecutive pair $A_t$, $A_{t+1}$ in the sequence be locally conjugate. Let $\{U_1, U_2, \ldots, U_{\tau}\}$ be a set of constant-depth unitaries that each fulfil the conditions listed in definition~\ref{def:genunitaries} for general logical unitaries in Floquet codes for the respective transitions $\Acurvy_t \rightarrow \Acurvy_{t+1}$. Then, the combined logical action of all unitaries applied at the individual time steps and the conjugate basis measurements inducing the transitions $\Acurvy_t \rightarrow \Acurvy_{t+1}$ is an element of the $D$-th level of the Clifford Hierarchy.
 \end{corollary}
 \begin{proof}
 	Theorem~\ref{thm:explocal} implies that for each pair $\Acurvy_t,\, \Acurvy_{t+1}$, the generalised logical unitary $U_t$ abiding to the conditions listed in definition~\ref{def:genunitaries} fulfils  Eqn.~\eqref{eq:Ugenrep}, i.e. $U_t\Pi_{\Acurvy_t}=e^{i  \vec{\phi} \vec{b} }U_{\Acurvy_t}\Pi_{\Acurvy_t}$, where $e^{i  \vec{\phi} \vec{b} }$ is constant-depth and $U_{\Acurvy_t}$ is unitary on $\Hil_{\Acurvy_t}$. Therefore, there exists the code-preserving unitary $U_{a,t} :=e^{-i  \vec{\phi} \vec{b} } U_t$, which, as a product of two constant-depth circuits, is constant-depth. Thus, the unitary map induced by $U$ and the measurement is given by $V_{t,t+1} U_{a,t} $, where $U_{a,t}$ is constant-depth and $V_{t,t+1}$ is a constant-depth circuit due to corollary~\ref{cor:Vablocal}. As $\tau =\mathcal{O}(1)$, we can then immediately apply the original BK theorem to the product of $\tau$ constant-depth circuits.
 \end{proof}
 
 \subsection{Proof of theorem~\ref{thm:explocal}}
 We will prove theorem~\ref{thm:explocal} by contradiction. We will do so by showing that - assuming that $U$ is constant-depth - the connected correlation between two observables $A$ and $B$ whose distance is super-constant must be zero. Then, we show that this is not the case when $e^{i\vec{\phi}\vec{b}}$ is not constant-depth.
 First, in lemma~\ref{lem:correlationstab}, we prove that the connected correlation between two observables whose respective supports are fully contained within two far apart constant regions is zero with respect to the ground state of a topological stabiliser code. In the following, we will abbreviate the support of an observable $H$ $\mathrm{supp}(H)$ by $|H|$.~ \begin{lemma}\label{lem:correlationstab}
 	Let $\Scurvy$ be a stabiliser group generated by the geometrically local generators $\{s_i\}$ with a distance $d$. Then, for any two observables $A$ and $B$ where $|A|,\, |B| <d$ and $\mathrm{dist}(|A|,|B|)= \omega(1)$ and therefore $|A| \cap |B| =\emptyset$, the connected correlation function with respect to any code state $\ket{\Psi} \in \Hil_{\Scurvy}$ is zero: 
 		\begin{equation}
 		\left|\bra{\Psi} A B \ket{\Psi} - 	\bra{\Psi} A \ket{\Psi}\bra{\Psi} B \ket{\Psi}\right| =0
 	\end{equation}
 \end{lemma}~\begin{proof}
 	We denote the normaliser of $\Scurvy$ by $\mathcal{N}(\Scurvy)$ or simply $\mathcal{N}$ as an abbreviation. For any $P \in \mathcal{P}_n$ and any $s \in \{s_i\}$, $P$ either commutes or anti-commutes with $s$. Thus, if $P \notin \mathcal{N} \setminus \Scurvy$, its expectation value with respect to any $\ket{\Psi} \in \Hil_{\Scurvy}$ is:
 	\begin{equation} \label{eq:experror0}
 		\bra{\Psi} P \ket{\Psi}=0,
 	\end{equation}
 for all $P \notin \mathcal{N} \setminus \Scurvy$. We write out both $A$ and $B$ in terms of the Pauli basis:
  \begin{equation}
 	A =\sum_{i} \alpha_{i} P_i, \label{eq:APauli}
 \end{equation}
and 
  \begin{equation}
 	B=\sum_{j} \beta_{j} P_j,\label{eq:BPauli}
 \end{equation}
 where each $P_i, P_j \in \mathcal{P}_n$. By definition, the support of each $P_i$ and $P_j$ in each sum in Eqn.~\eqref{eq:APauli} and Eqn.~\eqref{eq:BPauli}  fulfils:
  \begin{equation}
  	|P_i| \leq |A| = \mathcal{O}(1),\quad |P_j| \leq |B| = \mathcal{O}(1).
 \end{equation}
And thus, as per definition of distance, each $P_i$ ($P_j$) cannot be a non-trivial element of the normaliser, i.e. $P_i \notin \mathcal{N}\setminus \Scurvy$ ($P_j \notin \mathcal{N} \setminus \Scurvy$). They are all either a stabiliser: $P_i \in \Scurvy$ or fulfil Eqn.~\eqref{eq:experror0}. Thus, we can compute their expectation values:
 \begin{equation}
 	\bra{\Psi} A \ket{\Psi} = \sum_{i} \alpha_{i} \delta_{P \in \Scurvy},
 \end{equation}
 and 
    \begin{equation}
 	\bra{\Psi} A \ket{\Psi} = \sum_{j} \beta_{j} \delta_{P \in \Scurvy}.
 \end{equation}
 Additionally, as both $|A|$ and $|B|$ are constant, $|P_i P_j|$ has constant support and thus the expectation value $\bra{\Psi} P_i P_j \ket{\Psi}$ is either 0 if $P_iP_j \notin \Scurvy$ or 1 if $P_i P_j \in \Scurvy$.
 However, as $\mathrm{dist}(|A|,|B|)=\omega(1)$ and the stabilisers are locally generated, $P_i P_j$ can only be a stabiliser if both $P_i \in \Scurvy$ and $P_j \in \Scurvy$:
 \begin{equation}
 	\bra{\Psi} P_i P_j \ket{\Psi} = \delta_{P_i \in \Scurvy} \delta_{P_j \in \Scurvy}.
 \end{equation}
 And therefore: 
 \begin{align}
	\bra{\Psi} AB \ket{\Psi}&=\sum_{i,j} \alpha_{i}\beta_j \delta_{P_iP_j \in \Scurvy} \\
	&=\sum_{i,j} \alpha_{i}\beta_j \delta_{P_i \in \Scurvy} \delta_{P_j \in \Scurvy} \\
	&= \left(\sum_{i} \alpha_{i}\delta_{P_i \in \Scurvy}\right) \left(\sum_j \beta_j   \delta_{P_j \in \Scurvy}\right). \label{eq:expfactorises}
\end{align}
 Thus, the connected correlation for all $\ket{\Psi} \in \Hil_{\Scurvy}$ is:
 \begin{align}
C&= |	\bra{\Psi} A B \ket{\Psi}-\bra{\Psi} A  \ket{\Psi}\bra{\Psi} B  \ket{\Psi}|\\
&= \left(\sum_{i} \alpha_{i}\delta_{P_i \in \Scurvy}\right) \left(\sum_j \beta_j   \delta_{P_j \in \Scurvy}\right) \notag\\
&\hspace{10em}-\left(\sum_{i} \alpha_{i} \beta_j \delta_{P_i P_j \in \Scurvy}\right) \\
&=0,
 \end{align}
where we inserted Eqn.~\eqref{eq:expfactorises}.
 \end{proof}
Now we can continue with the proof of theorem~\ref{thm:explocal}. For any Hilbert space $\Hil$ and any unitary $U$, we introduce the short-hand notation $U\Hil U^{\dag}$ for the image of $\Hil$ under $U$.

\begin{proof}
	
First we assume, w.l.o.g., that all $\phi_b \neq  \pi k$, where $k$ is an integer. Such a $\phi_b$ only contributes a Pauli operator to $e^{i\vec{\phi} \vec{b}}$, i.e. a transversal operator. We can prove theorem~\ref{thm:explocal} for $e^{i\vec{\phi} \vec{b}}$ with all constant-depth layers removed.
	
If we assume $e^{i\vec{\phi} \vec{b}}$ is not a constant-depth circuit, then there exists at least one $b \in B_e$ such that $\mathrm{diam}(|b|)=\omega(1)$. That implies that there exist an $a$ and an $a'$, both in $\Acurvybar$, with $\mathrm{dist}(|a|,|a'|)=\omega(1)$ which both anti-commute with $b$. More precisely, we assume that $\mathrm{dist}(|a|,|a'|)>hr$, where $h$ and $r$ are the depth and range of $U$. (For which we assume $hr=\mathcal{O}(1)$.) Keeping this in mind, we will now show that due to lemma~\ref{lem:correlationstab} their connected correlation with respect to a $\ket{\Psi} \in U\Hil_{\Acurvy}U^{\dag}$ must be zero:
\begin{equation}\label{eq:Cproof}
	C=  |\bra{\Psi} a a'\ket{\Psi}-\bra{\Psi}  a'\ket{\Psi}\bra{\Psi}  a'\ket{\Psi}|=0.
\end{equation}
%
For this, let $H = U^{\dag} a U$. It support is bound to:
\begin{equation}
	|H|\subseteq \Bcurvy_{hr}(|a|).
\end{equation}
We write out $H$ in terms of Paulis:
\begin{equation}
	H = \sum_{P| |P|\subseteq B_{\mathcal{O}(1)}(|a|)} \alpha_P P.
\end{equation}
We analogously define $H'=U^{\dag}a'U$, which we can also write out in terms of local Paulis in the same manner as $H$. Note that, as $\mathrm{dist}(|a|,|a'|)= \omega(1)$ and $U^{\dag}$ maps both $a$ and $a'$ to local neighbourhoods, $|H|\cap |H'|=\emptyset$.
Now we insert $H$ and $H'$ into Eqn.~\eqref{eq:Cproof}:
\begin{align}
	C&=	|\bra{\Psi} a a'\ket{\Psi}-\bra{\Psi}  a'\ket{\Psi}\bra{\Psi}  a'\ket{\Psi}|\\
	&=|\bra{\Psi} U U^{\dag} a U U^{\dag}  a'UU^{\dag}\ket{\Psi}\notag\\
	&\hspace{2em} -\bra{\Psi} U U^{\dag} a' U U^{\dag}\ket{\Psi}\bra{\Psi} U U^{\dag}  a' U U^{\dag}\ket{\Psi}| \label{eq:insertid}\\
	&=|\bra{\Phi}  HH' \ket{\Phi}-\bra{\Phi}  H \ket{\Phi}\bra{\Phi}  H' \ket{\Phi}| \label{eq:defphi}\\
	&=0
\end{align}
where we inserted identity in line~\eqref{eq:insertid}, defined $\ket{\Phi}:= U^{\dag} \ket{\Psi}$ in line~\eqref{eq:defphi}, and utilised lemma~\ref{lem:correlationstab} in the last line. Crucially this works, as both $H$ and $H'$ are supported on non-overlapping regions of constant size. The image of $U\Hil_{\Acurvy}U^{\dag}$ under $U^{\dag}$ is once again $\Hil_{\Acurvy}$, meaning that $\ket{\Phi} \in \Hil_{\Acurvy}$. Lemma~\ref{lem:correlationstab} holds for all $\ket{\Phi} \in \Hil_{\Acurvy}$. Thus, we can conclude that $C=0$ for all $\ket{\Psi} \in U\Hil_{\Acurvy}U^{\dag}$. \newline
We will now show how this leads to a contradiction if we impose the decomposition $U=e^{i\vec{\phi} \vec{b}}U_{\Acurvy}$. As $U_{\Acurvy}$ is code preserving, we can find a $\ket{\Phi} \in \Hil_{\Acurvy}$ for any $\ket{\Psi} \in U\Hil_{\Acurvy}U^{\dag}$ such that:
\begin{equation}
	\ket{\Psi}=e^{i\vec{\phi} \vec{b}}\ket{\Phi}.
\end{equation}
Thus, computing the connected correlator $C$ again yields:
\begin{align}
C&=\left|\bra{\Psi} aa'\ket{\Psi} - \bra{\Psi}a\ket{\Psi}\bra{\Psi}a'\ket{\Psi}\right|\\
&=|\bra{\Phi}e^{-i\vec{\phi}\vec{b}} aa'e^{i\vec{\phi}\vec{b}}\ket{\Phi} \\
&\hspace{3em}- \bra{\Phi}e^{-i\vec{\phi}\vec{b}} a e^{i\vec{\phi}\vec{b}}\ket{\Phi}\bra{\Phi}e^{-i\vec{\phi}\vec{b}} a'e^{i\vec{\phi}\vec{b}}\ket{\Phi}|.
\end{align}
Note that for any $Q \in \mathcal{P}_n$:
\begin{align}
	Q^{-\phi} &:= e^{-i\vec{\phi}\vec{b}} Q e^{i\vec{\phi}\vec{b}}\\
	&=\prod_{b \in K(Q)} Q \left(\cos(2\phi_b) -i\sin(2\phi_b) b\right),
\end{align}
where $K(Q)$ is defined as:
\begin{equation}
	K(Q):= \{b \in B_e: \{b,Q\}=0\}.
\end{equation}
Here $\{.,.\}$ denotes the anti-commutator. From now on, we abbreviate:
\begin{equation}
	f(\phi_b):= \cos(2\phi_b),
\end{equation}
and additionally, in a slight abuse of notation, for any $Q \in \mathcal{P}_n$:
\begin{equation}
	f(Q):= \{ \cos(2\phi_b)| b \in K(Q)\}.
\end{equation}
Thus, as $\ket{\Phi} \in \Hil_{\Acurvy}$ and $a$ and $a'$ are stabilisers in $\Hil_{\Acurvy}$, whereas any $b \in B_e \subseteq \Bcurvybar$ yields $\bra{\Phi}b \ket{\Phi}=0 $ for all $\ket{\Phi} \in \Hil_{\Acurvy}$, we can directly compute:
\begin{align}
\bra{\Psi}a^{-\phi}\ket{\Psi}&=	\prod_{b \in K(a)} f(\phi_b)=\prod_{x \in f(a)} x, \\
\bra{\Psi}a'^{-\phi}\ket{\Psi}&=	\prod_{x \in f(a')} x, \\
	\bra{\Psi}aa'^{-\phi}\ket{\Psi}&=	\prod_{x \in f(aa')} x.
\end{align}
As $f(aa')$ can theoretically be empty, we define:
\begin{equation}
	\prod_{x \in \emptyset} x := 1.
\end{equation}
 In total:
	\begin{align}
		C&=|\bra{\Psi} (aa')^{-\phi}\ket{\Psi} - \bra{\Psi}a^{-\phi}\ket{\Psi}\bra{\Psi}a'^{-\phi}\ket{\Psi}|\\
		&= \left|\prod_{x \in f(aa')}x-\left(\prod_{x \in f(a)}x\right)\left(\prod_{ x' \in f(a')}x'\right) \right|\label{eq:fdoubleprod}\\
		&=\left|\prod_{x \in f(aa')}x-\prod_{ x' \in f(a) \sqcup f(a')} x'\right| \label{eq:fdisjointunion},
	\end{align}	
	where $\sqcup$ denotes the disjoint union. In a slight abuse of notation, we neglect the indexation in the disjoint union and treat it as a collection of sets. We let the product in Eqn.~\eqref{eq:fdisjointunion} run over all of these sets.
	The set $f(aa')$ is given by:	
	\begin{equation}
	f(aa')=	\{f(\phi_b)| \{b,a\}=0 \veebar \{b,a'\}=0\},
	\end{equation} 	
	where $\veebar$ denotes the \textbf{exclusive} OR.
 As there exists at least one $b \in B_e$ that anti-commutes with both $a$ and $a'$, the product over $f(aa')$ is not the same as product over the disjoint union in Eqn.~\eqref{eq:fdisjointunion}. Each $f(\phi_b)$ for a $b$ that exclusively anti-commutes with either $a$ or $a'$ appears in the product over $f(aa')$ \emph{and} in the product over the disjoint union in Eqn.~\eqref{eq:fdisjointunion}, but the latter product also runs over terms in the intersection of $f(a)$ and $f(a')$:
 \begin{equation}\label{eq:Cfinalform}
 	C=\left|\prod_{x \in f(aa')}x\right|\left|1-\prod_{x \in f(a) \cap f(a')}x^2\right|
 \end{equation}
 The power of two appears as each term in the double product in Eqn.~\eqref{eq:fdoubleprod} contributes one $x \in f(a) \cap f(a')$.
 As per assumption, all products in Eqn.~\eqref{eq:Cfinalform} run over  terms larger than 0 and strictly smaller than 1. Thus, $C$ cannot be zero, which yields the contradiction.

\end{proof}

\section{Discussion and outlook}\label{sec:concl}

We have analysed limitations for fault-tolerant gates in dynamical stabiliser codes. More specifically, we have discussed the implementation of logical operations via constant-depth circuits for topological Floquet codes, in which one switches between instantaneous stabiliser groups (ISGs) via local measurements. Our question was whether constant-depth circuits at different time steps composed with the measurements can circumvent the Bravyi-K\"onig theorem~\cite{bravyiClassificationTopologicallyProtected2013}.

To do so, we first fixed a definition of Floquet codes based on locally conjugate stabiliser groups~\cite{aasenMeasurementQuantumCellular2023a}, which is in line with many such codes introduced so far, such as those in ~\cite{hastingsDynamicallyGeneratedLogical2021,davydovaQuantumComputationDynamic2024,davydovaFloquetCodesParent2023,townsend-teagueFloquetifyingColourCode2023a}. This is however not the only possible definition. For instance, the authors of ~\cite{fuErrorCorrectionDynamical2025} utilise a broader definition of dynamical codes. They define a general sequence of stabiliser measurements as a dynamical code, and discuss error correction and detection in detail for this sequence. Additionally, they consider the possibility that at each time step the logical information is encoded in a subsystem code.\\

We summarised how conjugate stabiliser groups enable the preservation of logical information in section~\ref{sec:presinfo}.
We then showed that the original Bravyi-K\"onig theorem  holds for our definition of Floquet codes based on locally conjugate stabiliser groups and \emph{code preserving} constant-depth circuits at every time step. \\

We observed that the projective measurements allow for additional freedom in the choice of circuits applied at each time step. That is, one can apply unitaries that need not preserve the codespace but that nevertheless, combined with the subsequent projective measurement, constitute a valid logical operation. Our question was whether these generalised logical unitaries allow for circumventing the BK theorem. Therefore, in section~\ref{sec:beyondcodepres}, we first defined a sensible class of such generalised logical unitaries based on conditions we consider in line with the nature of Floquet codes. Then, we derived a canonical form. It turns out to be similar to the form of unitary operations in subsystem codes, such as the one considered in~\cite{pastawskiFaulttolerantLogicalGates2015}. However, the latter factorise with respect to the tensor product structure $\Hil_L \otimes \Hil_G$ of the subsystem code, where $\Hil_L$ is the logical space and $\Hil_G$ the gauge qubit space. Here, we do not require such factorisation. 

In section~\ref{sec:BKbeyondcodepres}, we showed the main result of this work, namely that the Bravyi-K\"onig theorem also holds for our generalised logical unitaries.
It should be noted that Fu et al. prove a Pastawski-Yoshida theorem (see~\cite{pastawskiFaulttolerantLogicalGates2015}) for two ISGs in a Floquet sequence in~\cite{fuErrorCorrectionDynamical2025}. Our work differs from theirs, as we do not view two ISGs $\Acurvy$ and $\Bcurvy$ in a Floquet sequence in a dynamical code as a subsystem code with $\Acurvy \cap \Bcurvy$ as the stabiliser group and $\Acurvy \cup \Bcurvy$ as the gauge group. The Pastawski-Yoshida theorem applies to unitaries that are independent of the gauge operators, which, as mentioned above, are different to the general unitaries we discuss. Additionally, we consider the combined logical action of unitaries at the different time steps, whereas the authors of~\cite{fuErrorCorrectionDynamical2025} consider one unitary at a time that may connect a stabiliser group $\Acurvy_t$ to some other $\Acurvy_{t'}$, where $t$ and $t'$ are different time steps, but they do not consider the combination of these unitaries with the measurements. \\

It would be interesting to further analyse limitations on fault tolerant operations for more general definitions of Floquet codes or spacetime codes, as in~\cite{fuErrorCorrectionDynamical2025,rodatzFloquetifyingStabiliserCodes2024a,delfosseSpacetimeCodesClifford2023}. The preservation of logical information due to the measurement of a Pauli operator that anti-commutes with at least one stabiliser element, as discussed in~\ref{sec:presinfo}, could be extended to general space-time codes, as defined in~\cite{delfosseSpacetimeCodesClifford2023}. That is because measuring a Pauli operator that anti-commutes with at least one stabiliser element amounts to measuring an element of the de-stabiliser. (See for example~\cite{aaronsonImprovedSimulationStabilizer2004a}). One could in principle always find conjugate bases for the destabiliser and the stabiliser and then replace the projective measurement by a unitary, similar to what we discuss in section~\ref{sec:presinfo}. This could be helpful to analyse the limitations of fault tolerant operations in spacetime codes.

As our generalised logical unitaries can also be viewed as operations on subsystem codes (though their action may depend on the state of the gauge qubits), one could assess what the statements made here imply for subsystem code gauge fixing.

We focused on the Bravyi-K\"onig theorem, which only holds for topological stabiliser codes, but everything should straightforwardly extend to disjointness style theorems~\cite{jochym-oconnorDisjointnessStabilizerCodes2018} or more general no-go theorems for bounded-spread quantum channels~\cite{websterUniversalFaulttolerantQuantum2022}, provided that any measurements in the Floquet sequence admit conjugate bases (as discussed in section~\ref{sec:revpairs}) that fulfil some abstract form of locality. By this we mean that limited growth of errors seems to generally imply non-universal computation. For dynamical codes, but also general space-time codes or subsystem code gauge fixing, it may be fruitful to utilise the replacement of measurements of a Pauli that anti-commutes with at least one stabiliser by unitaries (see section~\ref{sec:logicaleffectHH}), in order to quantify more clearly how much fault tolerance must be sacrificed in  general for universal computation.\\

Additionally, the fundamental working principle behind logical preservation in Floquet codes, the measurements that actually do not contain any information about the logical state, are not unique to Pauli measurements, but stem from mutually unbiased measurements~\cite{farkasMutuallyUnbiasedMeasurements2023}, which could be a way to define non-Pauli Floquet codes.

\section{Acknowledgments}
We thank Esther Xiaozhen Fu for the insightful comments and Alex Townsend-Teague and Johannes Jakob Meyer for the
for helpful discussions. JH acknowledges funding from the Dutch Research Council (NWO) through a Veni grant (grant No.VI.Veni.222.331). JH and JM acknowledge funding from the Quantum Software Consortium (NWO Zwaartekracht Grant No.024.003.037).

 \bibliography{bktex.bib}

@article{aaronsonImprovedSimulationStabilizer2004,
  title = {Improved Simulation of Stabilizer Circuits},
  author = {Aaronson, Scott and Gottesman, Daniel},
  year = 2004,
  month = nov,
  journal = {Physical Review A},
  volume = {70},
  number = {5},
  pages = {052328},
  publisher = {American Physical Society},
  doi = {10.1103/PhysRevA.70.052328},
  urldate = {2025-11-11}
}

@article{aaronsonImprovedSimulationStabilizer2004a,
  title = {Improved Simulation of Stabilizer Circuits},
  author = {Aaronson, Scott and Gottesman, Daniel},
  year = 2004,
  month = nov,
  journal = {Physical Review A},
  volume = {70},
  number = {5},
  pages = {052328},
  issn = {1050-2947, 1094-1622},
  doi = {10.1103/PhysRevA.70.052328},
  urldate = {2025-12-07},
  copyright = {http://link.aps.org/licenses/aps-default-license},
  langid = {english}
}

@misc{aasenMeasurementQuantumCellular2023a,
  title = {Measurement {{Quantum Cellular Automata}} and {{Anomalies}} in {{Floquet Codes}}},
  author = {Aasen, David and Haah, Jeongwan and Li, Zhi and Mong, Roger S. K.},
  year = 2023,
  month = aug,
  number = {arXiv:2304.01277},
  eprint = {2304.01277},
  primaryclass = {cond-mat, physics:math-ph, physics:quant-ph},
  publisher = {arXiv},
  doi = {10.48550/arXiv.2304.01277},
  urldate = {2024-09-02},
  archiveprefix = {arXiv}
}

@article{bombinUnifyingFlavorsFault2024a,
  title = {Unifying Flavors of Fault Tolerance with the {{ZX}} Calculus},
  author = {Bombin, Hector and Litinski, Daniel and Nickerson, Naomi and Pastawski, Fernando and Roberts, Sam},
  year = 2024,
  month = jun,
  journal = {Quantum},
  volume = {8},
  pages = {1379},
  publisher = {Verein zur F\"orderung des Open Access Publizierens in den Quantenwissenschaften},
  doi = {10.22331/q-2024-06-18-1379},
  urldate = {2026-01-10},
  langid = {british}
}

@article{bravyiClassificationTopologicallyProtected2013,
  title = {Classification of {{Topologically Protected Gates}} for {{Local Stabilizer Codes}}},
  author = {Bravyi, Sergey and K{\"o}nig, Robert},
  year = 2013,
  month = apr,
  journal = {Physical Review Letters},
  volume = {110},
  number = {17},
  pages = {170503},
  publisher = {American Physical Society},
  doi = {10.1103/PhysRevLett.110.170503},
  urldate = {2024-02-03}
}

@article{colladayRewiringStabilizerCodes2018,
  title = {Rewiring Stabilizer Codes},
  author = {Colladay, Kristina R and Mueller, Erich J},
  year = 2018,
  month = aug,
  journal = {New Journal of Physics},
  volume = {20},
  number = {8},
  pages = {083030},
  publisher = {IOP Publishing},
  issn = {1367-2630},
  doi = {10.1088/1367-2630/aad8dd},
  urldate = {2025-11-11},
  langid = {english}
}

@article{davydovaFloquetCodesParent2023,
  title = {Floquet {{Codes}} without {{Parent Subsystem Codes}}},
  author = {Davydova, Margarita and Tantivasadakarn, Nathanan and Balasubramanian, Shankar},
  year = 2023,
  month = jun,
  journal = {PRX Quantum},
  volume = {4},
  number = {2},
  pages = {020341},
  publisher = {American Physical Society},
  doi = {10.1103/PRXQuantum.4.020341},
  urldate = {2024-02-03}
}

@article{davydovaQuantumComputationDynamic2024,
  title = {Quantum Computation from Dynamic Automorphism Codes},
  author = {Davydova, Margarita and Tantivasadakarn, Nathanan and Balasubramanian, Shankar and Aasen, David},
  year = 2024,
  month = aug,
  journal = {Quantum},
  volume = {8},
  pages = {1448},
  publisher = {Verein zur F\"orderung des Open Access Publizierens in den Quantenwissenschaften},
  doi = {10.22331/q-2024-08-27-1448},
  urldate = {2026-01-10},
  langid = {british}
}

@misc{delfosseSpacetimeCodesClifford2023,
  title = {Spacetime Codes of {{Clifford}} Circuits},
  author = {Delfosse, Nicolas and Paetznick, Adam},
  year = 2023,
  month = may,
  number = {arXiv:2304.05943},
  eprint = {2304.05943},
  primaryclass = {quant-ph},
  publisher = {arXiv},
  doi = {10.48550/arXiv.2304.05943},
  urldate = {2025-11-11},
  archiveprefix = {arXiv}
}

@article{eastinRestrictionsTransversalEncoded2009,
  title = {Restrictions on {{Transversal Encoded Quantum Gate Sets}}},
  author = {Eastin, Bryan and Knill, Emanuel},
  year = 2009,
  month = mar,
  journal = {Physical Review Letters},
  volume = {102},
  number = {11},
  pages = {110502},
  publisher = {American Physical Society},
  doi = {10.1103/PhysRevLett.102.110502},
  urldate = {2024-02-21}
}

@misc{ellisonFloquetCodesTwist2023,
  title = {Floquet Codes with a Twist},
  author = {Ellison, Tyler D. and Sullivan, Joseph and Dua, Arpit},
  year = 2023,
  month = sep,
  number = {arXiv:2306.08027},
  eprint = {2306.08027},
  primaryclass = {quant-ph},
  publisher = {arXiv},
  doi = {10.48550/arXiv.2306.08027},
  urldate = {2024-02-05},
  archiveprefix = {arXiv}
}

@misc{ellisonFloquetCodesTwist2023a,
  title = {Floquet Codes with a Twist},
  author = {Ellison, Tyler D. and Sullivan, Joseph and Dua, Arpit},
  year = 2023,
  month = sep,
  number = {arXiv:2306.08027},
  eprint = {2306.08027},
  primaryclass = {quant-ph},
  publisher = {arXiv},
  doi = {10.48550/arXiv.2306.08027},
  urldate = {2024-04-17},
  archiveprefix = {arXiv}
}

@article{farkasMutuallyUnbiasedMeasurements2023,
  title = {Mutually {{Unbiased Measurements}}, {{Hadamard Matrices}}, and {{Superdense Coding}}},
  author = {Farkas, M{\'a}t{\'e} and Kaniewski, J{\k e}drzej and Nayak, Ashwin},
  year = 2023,
  month = jun,
  journal = {IEEE Transactions on Information Theory},
  volume = {69},
  number = {6},
  pages = {3814--3824},
  issn = {1557-9654},
  doi = {10.1109/TIT.2023.3269524},
  urldate = {2024-07-15}
}

@article{fuErrorCorrectionDynamical2025,
  title = {Error {{Correction}} in {{Dynamical Codes}}},
  author = {Fu, Esther Xiaozhen and Gottesman, Daniel},
  year = 2025,
  month = oct,
  journal = {Quantum},
  volume = {9},
  pages = {1886},
  publisher = {Verein zur F\"orderung des Open Access Publizierens in den Quantenwissenschaften},
  doi = {10.22331/q-2025-10-20-1886},
  urldate = {2026-01-10},
  langid = {british}
}

@misc{gottesmanIntroductionQuantumError2009,
  title = {An {{Introduction}} to {{Quantum Error Correction}} and {{Fault-Tolerant Quantum Computation}}},
  author = {Gottesman, Daniel},
  year = 2009,
  month = apr,
  number = {arXiv:0904.2557},
  eprint = {0904.2557},
  primaryclass = {quant-ph},
  publisher = {arXiv},
  doi = {10.48550/arXiv.0904.2557},
  urldate = {2024-04-15},
  archiveprefix = {arXiv}
}

@article{guFaulttolerantQuantumArchitectures2025,
  title = {Fault-Tolerant Quantum Architectures Based on Erasure Qubits},
  author = {Gu, Shouzhen and Retzker, Alex and Kubica, Aleksander},
  year = 2025,
  month = mar,
  journal = {Physical Review Research},
  volume = {7},
  number = {1},
  pages = {013249},
  issn = {2643-1564},
  doi = {10.1103/PhysRevResearch.7.013249},
  urldate = {2025-11-10},
  langid = {english}
}

@article{haahBoundariesHoneycombCode2022,
  title = {Boundaries for the {{Honeycomb Code}}},
  author = {Haah, Jeongwan and Hastings, Matthew B.},
  year = 2022,
  month = apr,
  journal = {Quantum},
  volume = {6},
  pages = {693},
  issn = {2521-327X},
  doi = {10.22331/q-2022-04-21-693},
  urldate = {2024-02-05},
  langid = {english}
}

@article{hastingsDynamicallyGeneratedLogical2021,
  title = {Dynamically {{Generated Logical Qubits}}},
  author = {Hastings, Matthew B. and Haah, Jeongwan},
  year = 2021,
  month = oct,
  journal = {Quantum},
  volume = {5},
  pages = {564},
  publisher = {Verein zur F\"orderung des Open Access Publizierens in den Quantenwissenschaften},
  doi = {10.22331/q-2021-10-19-564},
  urldate = {2024-02-03},
  langid = {british}
}

@article{jochym-oconnorDisjointnessStabilizerCodes2018,
  title = {Disjointness of {{Stabilizer Codes}} and {{Limitations}} on {{Fault-Tolerant Logical Gates}}},
  author = {{Jochym-O'Connor}, Tomas and Kubica, Aleksander and Yoder, Theodore J.},
  year = 2018,
  month = may,
  journal = {Physical Review X},
  volume = {8},
  number = {2},
  pages = {021047},
  publisher = {American Physical Society},
  doi = {10.1103/PhysRevX.8.021047},
  urldate = {2024-09-02}
}

@article{kesselringAnyonCondensationColor2024,
  title = {Anyon {{Condensation}} and the {{Color Code}}},
  author = {Kesselring, Markus S. and {Magdalena de la Fuente}, Julio C. and Thomsen, Felix and Eisert, Jens and Bartlett, Stephen D. and Brown, Benjamin J.},
  year = 2024,
  month = mar,
  journal = {PRX Quantum},
  volume = {5},
  number = {1},
  pages = {010342},
  publisher = {American Physical Society},
  doi = {10.1103/PRXQuantum.5.010342},
  urldate = {2024-04-19}
}

@article{pastawskiFaulttolerantLogicalGates2015,
  title = {Fault-Tolerant Logical Gates in Quantum Error-Correcting Codes},
  author = {Pastawski, Fernando and Yoshida, Beni},
  year = 2015,
  month = jan,
  journal = {Physical Review A},
  volume = {91},
  number = {1},
  pages = {012305},
  publisher = {American Physical Society},
  doi = {10.1103/PhysRevA.91.012305},
  urldate = {2024-04-15}
}

@misc{rodatzFloquetifyingStabiliserCodes2024,
  title = {Floquetifying Stabiliser Codes with Distance-Preserving Rewrites},
  author = {Rodatz, Benjamin and Po{\'o}r, Boldizs{\'a}r and Kissinger, Aleks},
  year = 2024,
  month = oct,
  number = {arXiv:2410.17240},
  eprint = {2410.17240},
  publisher = {arXiv},
  doi = {10.48550/arXiv.2410.17240},
  urldate = {2024-10-28},
  archiveprefix = {arXiv}
}

@misc{rodatzFloquetifyingStabiliserCodes2024a,
  title = {Floquetifying Stabiliser Codes with Distance-Preserving Rewrites},
  author = {Rodatz, Benjamin and Po{\'o}r, Boldizs{\'a}r and Kissinger, Aleks},
  year = 2024,
  month = dec,
  number = {arXiv:2410.17240},
  eprint = {2410.17240},
  primaryclass = {quant-ph},
  publisher = {arXiv},
  doi = {10.48550/arXiv.2410.17240},
  urldate = {2025-11-10},
  archiveprefix = {arXiv}
}

@article{townsend-teagueFloquetifyingColourCode2023,
  title = {Floquetifying the {{Colour Code}}},
  author = {{Townsend-Teague}, Alex and {de la Fuente}, Julio Magdalena and Kesselring, Markus},
  year = 2023,
  month = aug,
  journal = {Electronic Proceedings in Theoretical Computer Science},
  volume = {384},
  eprint = {2307.11136},
  primaryclass = {math-ph, physics:quant-ph},
  pages = {265--303},
  issn = {2075-2180},
  doi = {10.4204/EPTCS.384.14},
  urldate = {2024-02-04},
  archiveprefix = {arXiv}
}

@article{townsend-teagueFloquetifyingColourCode2023a,
  title = {Floquetifying the {{Colour Code}}},
  author = {{Townsend-Teague}, Alex and {de la Fuente}, Julio Magdalena and Kesselring, Markus},
  year = 2023,
  month = aug,
  journal = {Electronic Proceedings in Theoretical Computer Science},
  volume = {384},
  eprint = {2307.11136},
  primaryclass = {math-ph, physics:quant-ph},
  pages = {265--303},
  issn = {2075-2180},
  doi = {10.4204/EPTCS.384.14},
  urldate = {2024-02-14},
  archiveprefix = {arXiv}
}

@misc{vuillotPlanarFloquetCodes2021,
  title = {Planar {{Floquet Codes}}},
  author = {Vuillot, Christophe},
  year = 2021,
  month = dec,
  number = {arXiv:2110.05348},
  eprint = {2110.05348},
  primaryclass = {quant-ph},
  publisher = {arXiv},
  doi = {10.48550/arXiv.2110.05348},
  urldate = {2024-02-05},
  archiveprefix = {arXiv}
}

@article{websterUniversalFaulttolerantQuantum2022,
  title = {Universal Fault-Tolerant Quantum Computing with Stabilizer Codes},
  author = {Webster, Paul and Vasmer, Michael and Scruby, Thomas R. and Bartlett, Stephen D.},
  year = 2022,
  month = feb,
  journal = {Physical Review Research},
  volume = {4},
  number = {1},
  pages = {013092},
  publisher = {American Physical Society},
  doi = {10.1103/PhysRevResearch.4.013092},
  urldate = {2024-07-24}
}
 \onecolumngrid

\section{Remaining proofs}\label{sec:remainingproofs}

Here, we give the proofs for the statements left open in the main text.
First, we re-introduce some basic notation. Let $\Acurvy$ and $\Bcurvy$ be two conjugate stabiliser groups in a Floquet transition. Let $\Scurvy=\Acurvy\cap\Bcurvy$, $\Acurvybar = \Acurvy \setminus \Scurvy$ and $\Bcurvybar=\Bcurvy \setminus \Scurvy$. We denote the conjugate bases by $\{a_i\}$ and $\{b_i\}$, where $|\{a_i\}|=|\{b_i\}|=n_m$. For the joint eigenspace of a stabiliser group $\Gcurvy$ with generators $\{g_i\}$, we specify the eigenvalues by the vector $\vec{m}_{\Gcurvy} \in \mathbb{F}_2^{n_m}$,  where the eigenvalue of each respective $g_i$ is $(-1)^{m_{\Gcurvy,i}}$ (with respect to some enumeration). Further, the projection operator $\Pi_{\Gcurvy}(\vec{m}_{\Gcurvy})$ projects onto the joint eigenspace of $\{g_i\}$ determined by $\vec{m}_{\Gcurvy}$.  

We now prove proposition~\ref{prop:pbpapbidHH}, which we will restate here as lemma~\ref{lem:pbpapbid}:

\begin{lemma}\label{lem:pbpapbid}
	Let $\Acurvy$ and $\Bcurvy$ be two conjugate stabiliser groups. Then, for any $\vec{m}_{\Acurvybar}$ and any $\vec{m}_{\Bcurvybar}$:
	
	\begin{align}\label{eq:papbpaid}
			\Pi_{\Acurvybar}(\vec{m}_{\Acurvybar}) \Pi_{\Bcurvybar}(\vec{m}_{\Bcurvybar}) \Pi_{\Acurvybar}(\vec{m}_{\Acurvybar}) = \frac{1}{2^{n_m}} \Pi_{\Acurvybar}(\vec{m}_{\Acurvybar}).
	\end{align}
	and:
	
	\begin{align}\label{eq:pbpapbid}
		\Pi_{\Bcurvybar}(\vec{m}_{\Bcurvybar}) \Pi_{\overline{\Acurvy}}(\vec{m}_{\overline{\Acurvy}}) \Pi_{\Bcurvybar}(\vec{m}_{\Bcurvybar}) = \frac{1}{2^{n_m}} \Pi_{\overline{\Bcurvy}}(\vec{m}_{\overline{\Bcurvy}})
	\end{align}
	
\end{lemma}
 
\begin{proof}
First, the LHS of Eqn.~\ref{eq:papbpaid} factorises into a product of individual projection operators:
 
\begin{align}
		 	\Pi_{\Acurvybar}(\vec{m}_{\Acurvybar}) \Pi_{\Bcurvybar}(\vec{m}_{\Bcurvybar}) \Pi_{\Acurvybar}(\vec{m}_{\Acurvybar})&=	\prod_{i=0}^{n_m-1} \left(\frac{\mathds{1}+(-1)^{m^{(\overline{\Acurvy})}_i} a_i}{2} \right) \left(\frac{\mathds{1}+(-1)^{m^{(\overline{\Bcurvy})}_i} b_i}{2}\right)   \left(\frac{\mathds{1}+(-1)^{m^{(\overline{\Acurvy})}_i} a_i}{2}\right),
	\end{align} 
 where the index $i$ runs over an enumeration over the conjugate basis pairs, i.e. $a_i$ anti-commutes with exactly one generator of $\overline{\Bcurvy}$, namely $b_i$. This turns into:
 \begin{align}
		\phantom{=}&\prod_{i} \left(\frac{\mathds{1}+(-1)^{m^{(a)}_i} a_i}{2} \right) \left(\frac{\mathds{1}+(-1)^{m^{(b)}_i} b_i}{2}\right) \left(\frac{\mathds{1}+(-1)^{m^{(a)}_i} a_i}{2}\right)\\
		=&\frac{1}{2^{3n_m}}\prod_{i} \left(\mathds{1}+(-1)^{m^{(a)}_i} a_i+(-1)^{m^{(b)}_i} b_i+(-1)^{m^{(a)}_i+m^{(b)}_i} a_ib_i \right)   \left( \mathds{1}+(-1)^{m^{(a)}_i} a_i \right) \\
		&=\frac{1}{2^{2n_m}}\prod_{i}  \left(\mathds{1}+(-1)^{m_i^{(a)}} a_i \right) \\
		&= \frac{1}{2^{n_m}} \Pi_{\Acurvybar},
	\end{align}
 where we used the anti-commutation relations between $a_i$ and $b_i$.
Eqn.~\ref{eq:pbpapbid} is shown analogously.  
\end{proof}

\begin{lemma}
	Let $\Acurvy$ be a stabiliser group and $\mathcal{N}(\Acurvy)$ its normaliser. Let $L$ be a linear operator spanned by elements of $\mathcal{N}(\Acurvy)$. If $\Pi_{\Acurvy}L\Pi_{\Acurvy}$ is unitary on $\Hil_{\Acurvy}$, then there exists a unitary representative $U$ such that $U\Pi_{\Acurvy} =L \Pi_{\Acurvy}$. 
\end{lemma}

Below, we will leave out the specification of $\vec{m}_{\Acurvybar}$ and $\vec{m}_{\Bcurvybar}$. For completeness, we restate definition~\ref{def:KdefHH} as definition~\ref{def:KdefHHapp}.
  \begin{definition}[$K_{\Acurvy,\Bcurvy}$]\label{def:KdefHHapp}
 	Let $\Acurvy \leftrightarrow \Bcurvy$ be a reversible pair and let the system be in the state $\ket{\Psi} \in \Hil_{\Acurvy}$. Now project via $\Pi_{\Bcurvybar}$. The unitary operator $K_{\Acurvy,\Bcurvy}$ describes the effect of $\Pi_{\Bcurvybar}$:  
 	
 	\begin{equation}\label{eq:HHKapp}
 		K_{\Acurvy,\Bcurvy}: \Hil_{\Acurvy} \rightarrow \Hil_{\Bcurvy}: \ket{\Psi} \mapsto 2^{n_m/2} \Pi_{\Bcurvybar} \ket{\Psi}.
 	\end{equation}
 \end{definition}

Next, we will prove lemma~\ref{lem:KunitaryHH} in section~\ref{sec:floquettransitionop}, which we will restate here as lemma~\ref{lem:Kunitary}:
 \begin{lemma} \label{lem:Kunitary}
	$K_{\Bcurvy,\Acurvy}: \Hil_{\Bcurvy} \rightarrow \Hil_{\Acurvy}$ is the inverse of $K_{\Acurvy,\Bcurvy}$, i.e. $K_{\Bcurvy,\Acurvy}K_{\Acurvy,\Bcurvy}= \mathbb{1}_{\Acurvy}$, where $\mathbb{1}_{\Acurvy}$ denotes the identity on $\Hil_{\Acurvy}$. Additionally, $K_{\Bcurvy,\Acurvy}=K_{\Acurvy,\Bcurvy}^{\dag}$.
\end{lemma}
\begin{proof}
 First, we insert the definition of $K_{\Acurvy,\Bcurvy}$  and $K_{\Bcurvy,\Acurvy}$ (definition~\ref{def:KdefHHapp}) into $K_{\Bcurvy,\Acurvy}K_{\Acurvy,\Bcurvy}\ket{\Psi}$ for a $\ket{\Psi} \in \Hil_{\Acurvy}$ and utilise $\Pi_{\Bcurvy}=\Pi_{\Scurvy}\Pi_{\Bcurvybar}$:
  \begin{equation}
  		K_{\Bcurvy,\Acurvy}K_{\Acurvy,\Bcurvy}\ket{\Psi}=2^{n_m} \Pi_{\overline{\Acurvy}} \Pi_{\Bcurvy}\Pi_{\overline{\Bcurvy}} \Pi_{\Acurvy} \ket{\Psi} \\
  	=2^{n_m} \Pi_{\overline{\Acurvy}} \Pi_{\overline{\Bcurvy}}\Pi_{\Scurvy}\Pi_{\overline{\Bcurvy}} \Pi_{\Acurvy} \ket{\Psi} 
  	\end{equation}
  	
  	Next we rearrange the projection operators (note that $\Pi_{\Scurvy}$ commutes with $\Pi_{\Acurvybar}$ and $\Pi_{\Bcurvybar}$ and that $\ket{\Psi}$ is an eigenstate of $\Scurvy$ per assumption) make use of lemma~\ref{lem:pbpapbid} to obtain:
  	
  	\begin{equation}
  			K_{\Bcurvy,\Acurvy}K_{\Acurvy,\Bcurvy}\ket{\Psi}=2^{n_m} \Pi_{\overline{\Acurvy}} \Pi_{\overline{\Bcurvy}}\Pi_{\overline{\Bcurvy}} \Pi_{\Acurvy} \Pi_{\Scurvy}\ket{\Psi}=2^{n_m} \Pi_{\overline{\Acurvy}}  \Pi_{\overline{\Bcurvy}} \Pi_{\Acurvybar} \Pi_{\Acurvy} \ket{\Psi}=\Pi_{\Acurvybar} \Pi_{\Acurvy} \ket{\Psi} =\ket{\Psi}
  	\end{equation}
 To see that $K_{\Bcurvy,\Acurvy}=K_{\Acurvy,\Bcurvy}^{\dag}$, note that:
\begin{equation}
	\bra{\Psi}_{\Bcurvy} (2^{n_m}\Pi_{\Bcurvybar} \Pi_{\Acurvy}\ket{\Psi}_{\Acurvy})=(\bra{\Psi}_{\Bcurvy} 2^{n_m}\Pi_{\Bcurvy} \Pi_{\Acurvybar})\ket{\Psi}_{\Acurvy},
\end{equation}
for all $\ket{\Psi}_{\Acurvy} \in \Hil_{\Acurvy}$ and $\ket{\Psi}_{\Bcurvy} \in \Hil_{\Bcurvy}$.
\end{proof}

\end{document}